\documentclass{article}

\usepackage{graphicx} % Required for inserting images

\usepackage[english]{babel}
\usepackage[utf8]{inputenc}
\usepackage[T1]{fontenc}
\usepackage{bm}
\usepackage[a4paper, margin=1in]{geometry}
\usepackage{graphicx}
\usepackage{framed}
\usepackage[framemethod=tikz]{mdframed}
\usepackage{todonotes}
\usepackage{color}
\newcommand{\alp}{\alpha}
\newcommand{\eps}{\epsilon}
\newcommand{\Omg}{\Omega}
\newcommand{\lmax}{\ell_{\max}}

\newcommand{\taumax}{\tau_{\max}}
\newcommand{\calP}{\mathcal{P}}
\newcommand{\bbE}{\mathbb{E}}

\newcommand{\bundle}{\text{pass-bundle}\xspace}
\newcommand{\limit}{\mathsf{limit}}
\newcommand{\lab}{\mathsf{distance}}
\newcommand{\algPhase}{\textsc{Alg-Phase}\xspace}
\newcommand{\algExtend}{\textsc{Extend-Active-Path}\xspace}
\newcommand{\algBacktrack}{\textsc{Backtrack-Stuck-Structures}\xspace}

\newcommand{\algOvertake}{\textsc{Overtake}\xspace}
\newcommand{\algAugment}{\textsc{Augment}\xspace}
\newcommand{\algContract}{\textsc{Contract}\xspace}
\newcommand{\algCheck}{\textsc{Contract-and-Augment}\xspace}

\newcommand{\Aweak}{\mathbb{A}_{\mathsf{weak}}}

\newcommand{\Amat}{\mathbb{A}_{\mathsf{matching}}}
\newcommand{\Tmat}{\mathcal{T}_{\mathsf{matching}}}

\newcommand{\Aexp}{\mathbb{A}_{\mathsf{explore}}}

\newcommand{\Aproc}{\mathbb{A}_{\mathsf{process}}}
\newcommand{\Tproc}{\mathcal{T}_{\mathsf{process}}}

\newcommand{\OMv}{\mathsf{OMv}}
\newcommand{\paren}[1]{\left(#1\right)}
\newcommand{\ORS}[2]{\textsf{ORS}\paren{#1,#2}\xspace}
\newcommand{\rs}{Ruzsa-Szemerédi\xspace}

\newcommand\cev[1]{\overleftarrow{#1}}

\definecolor{darkgreen}{rgb}{0,0.5,0}
\definecolor{darkgray}{rgb}{0.2,0.2,0.2}
\usepackage{hyperref}
\hypersetup{
    unicode=false,          % non-Latin characters in Acrobat’s bookmarks
    colorlinks=true,        % false: boxed links; true: colored links
    linkcolor=blue,         % color of internal links (change box color with linkbordercolor)
    citecolor=purple,       % color of links to bibliography
    filecolor=magenta,      % color of file links
    urlcolor=cyan           % color of external links
}

\usepackage{amsthm}
\usepackage{amsmath}
\usepackage{amsfonts}
\usepackage{mathrsfs}
\usepackage{mathtools}
\usepackage{verbatim}
\usepackage{footnote}

\usepackage{algorithm}
\usepackage[noend]{algpseudocode}
\usepackage{lineno}
\usepackage{caption}
\usepackage{framed}
\usepackage{enumerate}

\usepackage{tikzsymbols}
\usepackage{thmtools,thm-restate}
\usepackage{nicefrac}

\usepackage{subcaption}
\usepackage{pdfpages}
\usepackage{makecell}

\usepackage[capitalize, nameinlink]{cleveref}
\usepackage[section]{placeins}
\crefname{theorem}{Theorem}{Theorems}
\Crefname{lemma}{Lemma}{Lemmas}
\Crefname{invariant}{Invariant}{Invariants}
\Crefname{claim}{Claim}{Claims}
\Crefname{observation}{Observation}{Observations}
\Crefname{algorithm}{Algorithm}{Algorithms}
\Crefname{figure}{Figure}{Figures}
\Crefname{challenge}{Challenge}{Challenges}

\newtheorem{theorem}{Theorem}[section]
\newtheorem{lemma}[theorem]{Lemma}
\newtheorem{corollary}[theorem]{Corollary}
\newtheorem{definition}[theorem]{Definition}
\newtheorem{invariant}[theorem]{Invariant}

\newtheorem{claim}[theorem]{Claim}
\newtheorem{remark}{Remark}
\newtheorem*{remark*}{Remark}
\newtheorem{problem}{Problem}

\DeclareMathOperator{\poly}{poly}

\def\CONGEST{\ensuremath{\mathsf{CONGEST}}\xspace}
\def\MPC{\ensuremath{\mathsf{MPC}}\xspace}

\newcommand{\peps}{\poly(1/\eps)}

\newcommand{\tO}{\tilde{O}}

\newcommand{\cO}{\mathcal{O}}
\newcommand{\cS}{\mathcal{S}}
\newcommand{\cX}{\mathcal{X}}
\newcommand{\cT}{\mathcal{T}}
\newcommand{\bbA}{\mathbb{A}}

\newcommand{\AOMv}{\mathbb{A_{\OMv}}}

\title{A framework for boosting matching approximation: \\
    parallel, distributed, and dynamic}
\author{ Slobodan Mitrović\thanks{Supported by the Google Research Scholar and NSF Faculty Early Career Development Program No.~2340048. e-mail: \texttt{smitrovic@ucdavis.edu}} \\ UC Davis 
\and Wen-Horng Sheu\thanks{Supported by the NSF Faculty Early Career Development Program No.~2340048. e-mail: \texttt{wsheu@ucdavis.edu}} \\ UC Davis}
\date{}

\begin{document}

\maketitle
\begin{abstract}
    This work designs a framework for boosting the approximation guarantee of maximum matching algorithms. 
    As input, the framework receives a parameter $\epsilon > 0$ and an oracle access to a $\Theta(1)$-approximate maximum matching algorithm $\mathcal{A}$.
    Then, by invoking $\mathcal{A}$ for $\text{poly}(1/\epsilon)$ many times, the framework outputs a $1+\epsilon$ approximation of a maximum matching. 
    Our approach yields several improvements in terms of the number of invocations to $\mathcal{A}$:
    \begin{itemize}
        \item In MPC and CONGEST, our framework invokes $\mathcal{A}$ for $O(1/\epsilon^7 \cdot \log(1/\epsilon))$ times, substantially improving on $O(1/\epsilon^{39})$ invocations following from [Fischer et al., STOC'22] and [Mitrovic et al., arXiv:2412.19057].
        \item In both online and offline fully dynamic settings, our framework yields an improvement in the dependence on $1/\epsilon$ from exponential [Assadi et al., SODA25 and Liu, FOCS24] to polynomial.
    \end{itemize}
\end{abstract}

% \newpage

% \tableofcontents

% \newpage
\section{Introduction}
Given a graph $G = (V, E)$, a matching $M \subseteq E$ is a set of edges that do not have endpoints in common.
A matching is called maximal if it is not a proper subset of any other matching, and it is called maximum if it has the largest cardinality among all matchings.
Computing maximum matching in polynomial time has been known since the 1960s~\cite{edmonds1965maximum,edmonds1965paths,micali1980v,gabow1990data}.
Moreover, in recent breakthroughs~\cite{chen2022maximum,van2023deterministic}, it was shown how to compute maximum matching in bipartite graphs in almost linear time.

In some settings, it is known that computing a maximum matching is highly inefficient, e.g., in semi-streaming~\cite{FeigenbaumKMSZ05,guruswami2016superlinear}, LOCAL, \CONGEST, dynamic~\cite{henzinger2015unifying}, and sublinear~\cite{parnas2007approximating,BehnezhadRR23}.
Moreover, approximate maximum matchings usually admit much simpler solutions, e.g., a textbook example is a $2$-approximate maximum matching that results from maximal matching.
This inspired a study of $(1+\eps)$-approximate maximum matchings where, given a parameter $\eps > 0$, the task is to find a matching whose size is at least $1/(1+\eps)$ times the size of a maximum matching.

This direction has seen a proliferation of ideas; we discuss some of this work and reference a long list of results on this topic in \cref{sec:related-work}.
This setting was addressed from two perspectives: designing a standalone $(1+\eps)$-approximate algorithm and developing an approximation \emph{boosting framework}. 
The latter refers to algorithms that, as input, receive access to an $r$-approximate maximum matching procedure. This procedure is then adaptively invoked multiple times to obtain a $1+\eps$ approximation.
Some examples of such frameworks, or reductions, include~\cite{mcgregor2005finding,BhattacharyaKS23} for unweighted and \cite{stubbs2017metatheorems,gamlath2019weighted,Bernstein2021,Bernstein2025} for weighted graphs.

Despite this extensive work, one of the central questions still remains open:
\emph{What techniques yield to very efficient boosting frameworks applicable to various computational models?}
Our work contributes to this growing body of research by designing a new boosting framework for dynamic matchings and improving the complexity of the existing boosting framework for static matchings.

\subsection{Our results}
\label{sec:our-results}
In the rest, we use MCM to refer to maximum cardinality matching, i.e., to maximum unweighted matching.
Our first contribution is a framework for the static setting that reduces the computation of $(1+\eps)$-approximate MCM to only $O(\log(1/\eps) / \eps^7)$ invocation of a $\Theta(1)$-approximate MCM.

\begin{restatable}{theorem}{frameworkdistr} \label{thm:framework-1}
Let $\Amat$ be an algorithm that returns a $c$-approximate maximum matching for a given graph $H$, where $c > 1$ is a constant.
Let $\Aproc$ be an algorithm that simultaneously exchanges small messages between the vertices of a component, and does that in time $\Tproc$ for any number of disjoint components of $G$ each of size at most $1/\eps^d$, where $d$ is a fixed constant. 
Then, there is an algorithm that computes a $(1 + \eps)$-approximate maximum
matching in $G$ in time $O((\Tmat + \Tproc) \cdot \eps^{-7} \cdot \log (1/\eps))$. Furthermore, the algorithm requires access to $\peps$ words of memory per each vertex.
\end{restatable}
The algorithm $\Aproc$ in \cref{thm:framework-1} is model-specific and simply keeps the vertices in the graph up-to-date. For instance, implementing $\Aproc$ in MPC takes $O(1)$ rounds as long as $1/\eps^d$ fits into the memory of a machine.
In CONGEST, \cite{FMU22} instantiates $\Aproc$ by choosing one representative vertex in each component. All messages within a component are delivered via the representative vertex. Since each component is connected and has $\poly(1/\eps)$ vertices, this operation is done within $\poly(1/\eps)$ rounds.
We comment more on this in \cref{sec:app-framework-distr}. 

In MPC, $\Amat$ can be implemented using \cite{GU19}'s algorithm. The algorithm assumes that the vertices and edges of the input graph are distributed across the machines. It computes a desired matching in $O(\sqrt{\log n})$ rounds.
In \cref{table:times-static} we outline the improvement our result leads in some settings.

\begin{table}[h]
\centering
\small
\setlength{\extrarowheight}{4pt}
\begin{tabular}{|c|c|c|c|}
\hline
Reference & Complexity in $\eps$ & Setting \\ \hline
\cite{FMU22} & $O(1/\eps^{52})$  & MPC \\ \hline
%\cite{mcgregor2005finding} & $\exp(1/\eps)$  & Semi-streaming & No \\ \hline
\cite{FMU22}  + \cite{MMSS25} & $O(1/\eps^{39})$  & MPC \\ \hline
\textbf{this work} -- corollary of \cref{thm:framework-1} & $O(1/\eps^{7} \cdot \log(1/\eps))$  & MPC \\ \hline
\cite{FMU22} & $O(1/\eps^{63})$  & \CONGEST \\ \hline
\cite{FMU22}  + \cite{MMSS25} & $O(1/\eps^{42})$  & \CONGEST \\ \hline
\textbf{this work} -- corollary of \cref{thm:framework-1} & $O(1/\eps^{10} \cdot \log(1/\eps))$  & \CONGEST \\ \hline
\end{tabular}

\caption{\label{table:times-static} An overview of the most related frameworks in the static setting. All the frameworks apply to general unweighted graphs. 
Additional results are referenced in \cref{sec:related-work}.
}
\end{table}

While the framework developed corresponding to \cref{thm:framework-1} can be applied in a static setting by increasing the overall execution time by only a $\poly(1/\eps)$ factor, it is not known how to utilize this in the dynamic setting without significantly increasing the update time.
On a high level, the reason is that \cref{thm:framework-1} requires access to an oracle that outputs an approximate maximum matching in an adaptively chosen graph.
However, even to describe such a graph $G$ takes $\Omega(|E(G)|)$ time, yielding inefficient approximation boosting for dynamic algorithms.

Nevertheless, a more efficient, but also \emph{more restrictive}, oracle can be implemented for dynamic matchings. 
Given a graph $G = (V, E)$, consider an oracle $\cO$ that for a given $U \subseteq V$ returns a $\Theta(1)$-approximate maximum matching in $G[U]$ when $G[U]$ has a large matching.
Observe that $\cO$ allows changing a subset of $U \subseteq V$, but the graph $G$ remains the same!
Moreover, the algorithm invoking $\cO$ is allowed to spend $n \cdot \poly(1/\eps)$ time preparing the desired subsets $U$; this is in line with \cite[Proposition 2.2]{assadi2024improved}.
Prior works \cite{assadi2024improved,BhattacharyaKS23} already defined this oracle, and showed how to use it to compute a $1+\eps$ approximation of maximum matching by paying an exponential dependence on $1/\eps$ in the running time. 
This motivates the following questions:
\emph{Can we design a framework that finds a $1+\eps$ approximation of a maximum matching in $G$ by invoking $\cO$ for $\poly(1/\eps)$ times?}
We answer this question in the affirmative.
\begin{theorem}[Informal version of \cref{thm:framework}]\label{main:dynamic}
Let $G = (V, E)$ be an $n$-vertex fully dynamic graph that starts empty and throughout, never has more than $m$ edges.
Let $\Aweak$ be an algorithm that, given a vertex subset $S \subseteq V(G)$ and a parameter $\delta$, returns in $\cT(n, m,\delta)$ time a matching of size at least $\lambda\delta n$ if the maximum matching in $G[S]$ is at least $\delta n$.
Then, there is an algorithm for the fully dynamic $(1+\eps)$-approximate matching problem with amortized update time $O(\cT(n, m,\poly \eps) / n \cdot \poly(\log n / \eps))$.
\end{theorem}

As an example, \cite{assadi2024improved} instantiates $\Aweak$ as an algorithm that, given access to the adjacency matrix, computes the desired output in $O(mn^{3\gamma \log (n)}/d)$ time, where $\gamma$ is a parameter controlling the approximation factor and $d$ is the maximum degree of the subgraph induced by vertices of the output matching.

\begin{table}[h]
\centering
\small
\setlength{\extrarowheight}{4pt}
\begin{tabular}{|c|c|c|c|}
\hline
Reference & Complexity in $\eps$ & Complexity in $n$ & Setting \\
\hline
\cite{BG24} & $(1/\eps)^{O(1/\eps)}$ & $\sqrt{n^{1+O(\eps)} \cdot \ORS{n}{\Theta_\eps(n)}}$ & dynamic \\
\hline
\cite{assadi2024improved} & $(1/\eps)^{O(1/(\eps\beta))}, \beta > 0$ & $n^\beta \cdot \ORS{n}{\Theta_{\beta,\eps}(n)}$ & dynamic \\ 
\hline
\cite{Liu24} & $\peps$ & $n / 2^{\Omg(\sqrt{\log n})}$ & dynamic, bipartite \\ 
\hline
\cite{Liu24} & $\peps$ & $n^{0.58}$ & offline dynamic, bipartite \\
\hline
\textbf{this work}, \cref{thm:AKS-dynamic} & $(1/\eps)^{O(1/\beta)}, \beta > 0$ & $n^\beta \cdot \ORS{n}{\Theta_{\beta,\eps}(n)}$ & dynamic \\ 
\hline
\textbf{this work}, \cref{thm:fast-dynamic} & $\peps$ & $n / 2^{\Omg(\sqrt{\log n})}$ & dynamic \\ 
\hline
\textbf{this work}, \cref{thm:offline} & $\peps$ & $n^{0.58}$ & offline dynamic \\
\hline
\end{tabular}

\caption{\label{table:times-dynamic} An overview of algorithms for fully dynamic $(1+\eps)$-approximate maximum matching that are based on the boosting framework of \cite{mcgregor2005finding} (see \cite{BhattacharyaKS23} for its adaptation in the dynamic setting).
\cite{assadi2024improved}'s algorithm is parameterized by a real number $\beta > 0$ and has an amorized update time of $(1/\eps)^{O(1/(\eps\beta))} \cdot n^\beta \cdot \ORS{n}{\Theta_{\beta,\eps}(n)}$.
\cref{thm:AKS-dynamic} improves this result to $(1/\eps)^{O(1/\beta)} \cdot n^\beta \cdot \ORS{n}{\Theta_{\beta,\eps}(n)}$.
This result shows an improved trade-off between dependence on $1/\eps$ and on $n$:
The exponent $\beta$ in $n^\beta$ can be made arbitrarily small, but at the expense of an increased dependence on $1/\eps$.
For any constant $\beta$, the dependence on $1/\eps$ is polynomial.
Note that the result of \cite{Liu24} and \cite{assadi2024improved} are incomparable, as the exact value of $\ORS{\cdot}{\cdot}$ remains unknown.}
\end{table}

\subsection{Related work}
\label{sec:related-work}
Several lines of work have studied approaches for boosting matching guarantees, e.g., improving approximation or obtaining weighted from unweighted matchings.
In the rest, we use MWM to refer to maximum weighted matching.

\paragraph{Frameworks for unweighted matchings.}
Among the frameworks for boosting matching approximation in unweighted graphs, perhaps the most influential is the one by McGregor~\cite{mcgregor2005finding}, initially designed to compute $1+\eps$ approximation in the semi-streaming setting.
This approach turns out to be robust enough so that, with appropriate modifications, it can be applied in the MPC~\cite{czumaj2018round,assadi2019coresets,ghaffari2022massively} and dynamic setting~\cite{BhattacharyaKS23,assadi2024improved}.
Applying this technique to dynamic graphs -- while not substantially increasing the running time complexity -- is significantly trickier than applying it in the static setting. Details on why this is the case are discussed in \cref{sec:framework-dynamic}.
Given that \cite{mcgregor2005finding} has an exponential dependence on $1/\eps$, all the techniques derived from it have at least a single-exponential dependence on $1/\eps$ as well. 
Our work yields a polynomial dependence in the dynamic setting, improving the exponentially priorly known.

In recent work, Fischer, Mitrović, and Uitto~\cite[arXiv version]{FMU22} proposed a framework that obtains a $1+\eps$ approximation of MCM by $\poly(1/\eps)$ times invoking an algorithm that computes a $\Theta(1)$ approximation.
This work applies to static graphs in several models, including LOCAL, CONGEST, and MPC.

As was observed in \cite{assadi2024improved,BhattacharyaKS23,Liu24}, some other approaches, such as \cite{ahn2013linear,assadi2021auction} for bipartite and \cite{tirodkar2018deterministic} for general graphs, can also be used in developing boosting frameworks.

\paragraph{Frameworks for weighted matchings.}
For computing weighted matchings, Gupta and Peng~\cite{gupta2013fully} provide a reduction from general weights to integer weights in the range $[1, \exp(O(1/\eps))]$. The result was presented in the context of dynamic matching, although it can be applied to other settings as well, e.g., semi-streaming~\cite{huang20221}.

Stubbs and Williams~\cite{stubbs2017metatheorems} develop a reduction from dynamic weighted to dynamic unweighted maximum matching. Namely, they show how to design a dynamic algorithm for $(2 + \eps) \alpha$-approximate weighted from a dynamic algorithm for $\alpha$-approximate unweighted maximum matching at only poly-logarithmic increase in the update time.

Gamlath, Kale, Mitrović, and Svensson~\cite{gamlath2019weighted} show how to reduce the computation of $1+\eps$-approximate weighted maximum matching in general graphs to $1+\eps$-approximate unweighted maximum matching in bipartite graphs. This reduction applies to the static setting in semi-streaming and MPC, and has an exponential dependence on $1/\eps$.

For the case of bipartite graphs, Bernstein, Dudeja, and Langley~\cite{Bernstein2021} develop a framework that reduces the task of computing fully dynamic $(1+\eps) \alpha$-approximate MWM to the task of computing fully dynamic $\alpha$-approximate MCMs. 
This reduction incurs a logarithmic in $n$ and exponential in $1/\eps$ overhead in the running time.
In the same work, the authors also develop reductions for general graphs, with approximation guarantees of $3/2 + \eps$ and $2 + \eps$.

In a very recent work, Bernstein, Chen, Dudeja, Langley, Sidford, and Tu~\cite{Bernstein2025} made a significant contribution. 
In the context of fully dynamic $(1+\eps)$-approximate MWM, they provide a reduction from a graph with weights in the range of $\poly(n)$ to graphs with weights in the range of $\poly(1/\eps)$.
This reduction incurs only $\poly(1/\eps)$ additive time.
Combined with \cite{Bernstein2021}, this results in the fully dynamic $(1+\eps)$-approximate MWM in bipartite graphs with only polynomial dependence on $1/\eps$ in the running time.

\paragraph{Other related work.}
The approximate maximum matching problem has been extensively studied in numerous settings.
For a list of such works, we refer a reader to these and references therein: (semi-)streaming \cite{ahn2011laminar,ahn2013linear,kapralov2013better,KapralovKS14,BuryS15,AssadiKLY16,AssadiKL17,ahn2018access,EsfandiariHLMO18,BuchbinderST19,KapralovMNT20,gamlath2019weighted,AssadiB21,ChenKPS0Y21,Kapralov21,huang2023,AssadiS23,assadi2024simple}, MPC \cite{lattanzi2011filtering,czumaj2018round,GGK+18,assadi2019coresets,BBD+19,behnezhad2019exponentially,GU19,GGJ20,ghaffari2022massively,dhulipala2024parallel}, \CONGEST and LOCAL \cite{czygrinow2004fast,lotker2015improved,ahmadi2018distributed,bar2017distributed,harris2019distributed,fischer2020improved,ghaffari2023faster,izumi2024nearly}, and dynamic~\cite{GLSSS19,BGS20, ABD22,AssadiBKL23,BlikstadK23,ZhengH23,BhattacharyaKS23,BG24,assadi2024improved}.

\subsection{Paper organization}
The paper is organized as follows. \cref{sec:prelim} reviews standard notations;
\cref{sec:simp} describes a simplified version of \cite{MMSS25}'s algorithm;
\cref{sec:framework-distr} gives a boosting framework faster than \cite{FMU22}'s framework, with applications in $\MPC$ and $\CONGEST$;
\cref{sec:framework-dynamic} adapts the framework from \cref{sec:framework-distr} to the dynamic $(1+\eps)$-matching problem;
\cref{sec:app-framework-distr} provides additional implementation details.
\section{Overview of our approach}
\paragraph{Starting point.}
Our result is inspired by the framework of \cite{FMU22}.
First, that work describes an algorithm for computing $(1+\eps)$-approximate maximum matching in semi-streaming in $\poly(1/\eps)$ passes.
Second, that algorithm is extended to a framework that gets an oracle access to: (1) a method for computing $\Theta(1)$-approximate maximum matchings, and (2) a few simple-to-implement methods on graphs, such as exploring a local neighborhood of a vertex of size $\poly(1/\eps)$.\footnote{The framework also requires an access to $\Omega(n)$ space. This space can be distributed as well.}
Then, invoking these methods on $\poly(1/\eps)$ adaptively chosen graphs, the framework outputs a $1+\eps$ approximation of a maximum matching.

That framework is applicable in any setting that can provide oracle access to those methods. In particular, regarding the dependence on $1/\eps$, \cite{FMU22} obtain new results in the static setting in semi-streaming, MPC, \CONGEST.

\paragraph{Our improvement in the static setting.}
The same as \cite{FMU22}, our approach also starts with an algorithm that is not a framework. 
In our case, it is the algorithm of \cite{MMSS25}. Plugging \cite{MMSS25} directly into \cite{FMU22} already gives an improved framework using $O(1/\eps^{39})$ oracle calls to $\Theta(1)$-approximate maximum matchings; \cite{FMU22} performs $O(1/\eps^{52})$ many calls.
Our algorithm suffices to perform only $O(\log(1/\eps) / \eps^{7})$ such calls!
To achieve that, we significantly improve the efficiency of the \cite{FMU22}'s framework. 
We now briefly outline those changes, while details are presented in \cref{sec:framework-distr}.
\cref{sec:simp} outlines the algorithm of \cite{MMSS25}.

The main bottleneck of \cite{FMU22}'s framework is the simulation of two procedures.
The first procedure can be formulated as the following matching problem:
A graph $H$ is given.
In each iteration, we find a $c$-approximate matching in $H$ and remove all matched vertices.
The process is repeated until the maximum matching size of $H$ drops below a threshold $t$, from an initial value of $s$.
It is not hard to show that the process requires at most $(s - t) / (t / c)$ iterations, as each iteration finds a matching of size at least $t / c$.
We made a simple observation that, in fact, the maximum matching size in $H$ is decreasing exponentially, and thus $\Theta_c(\log\frac{s-t}{t})$ suffices.
This observation enables a simulation using $\Theta(\log (1/\eps))$ calls to the $c$-approximate matching algorithm, instead of $\peps$ calls.

Simulating the second procedure can be formulated as a more complicated matching problem.
We again are given a graph $H$ and aim to decrease its maximum matching size to below a threshold.
In each iteration, we find a $c$-approximate matching in $H$, but now only the matched edges, and not the vertices, are removed from $H$.
In addition, depending on the state of the algorithm, new edges may be added to $H$ after an iteration.
Therefore, our previous observation is not applicable.
To obtain our result, we present a different simulation for the procedure.
Roughly speaking, we show that the edges in $H$ can be divided into $\Theta(\eps^{-1})$ different classes, and the simulation for different classes of edges can be done separately.
In addition, the simulation of each class can be formulated as a matching problem similar to the first procedure.
Utilizing our previous observation, we show that each class requires $\Theta(\log (1/\eps))$ calls to the $c$-approximate matching algorithm, yielding a simulation of $O(\eps^{-1} \cdot \log(1/\eps))$ calls.
Then, our advertised complexity of $O(\log(1/\eps) / \eps^7)$ calls follows from the fact that our simulation follows -- a slightly simplified version of -- the algorithm of \cite{MMSS25}, that has $1/\eps^6$ dependence.

\paragraph{Our extension to the dynamic setting.}
As already discussed in \cref{sec:our-results}, recent results for dynamic matching provide access to an oracle which, given a subset of vertices $U \subseteq V$, outputs a $\Theta(1)$-approximate maximum matching in $G[U]$ if $G[U]$ contains a large matching.
Both the framework of \cite{FMU22} and our framework for the static setting require access to an oracle that outputs a $\Theta(1)$-approximate maximum matching in an adaptively chosen graph.
On a very high level, these two frameworks maintain so-called structures from each unmatched vertex. For the purpose of this discussion, a structure corresponding to an unmatched vertex $\alpha$ can be thought of as a set of alternating paths originating at $\alpha$. Structures corresponding to different vertices are vertex-disjoint.
Each structure attempts to extend an alternating path it contains. When an augmentation involving $\alpha$ is found, the entire structure corresponding to $\alpha$ is removed from the graph.
These extensions and augmentations in the static setting are handled by defining an appropriate graph $H$, and then finding a large matching in $H$. 
The way $H$ is defined, the sets $V(H)$ and $E(H)$ change as structures change from step to step, even if the set of unmatched vertices remains the same.
Each matching edge in $H$ is then mapped back to an extension or an augmentation.\footnote{This is a simplified exposition of the actual process. In the full algorithm, more operations are performed, but they follow the same logic we present here.
}
However, in the dynamic case, we do not have access to an oracle that can compute a matching in $H$.

Our first observation is that we do not need to find all the edges in $H$ affecting $\alpha$'s structure, but finding one such edge already makes progress. 
That is, if we know a priori that the alternating path $P$ of a structure will be extended, we could look for an extension of the head vertex of $P$ only.
This now allows ideas of randomly sampling a vertex from a structure, and hoping that a sampled vertex is ``the right'' one to perform an extension on. 
Indeed, that is exactly what our approach does.

However, this sampling idea does not suffice.
The reason is that our static framework maintains two types of vertices -- inner and outer ones.
Depending on whether a matched edge in $H$ is an outer-inner or outer-outer vertex determines which procedure is invoked to update the state of our framework.
In particular, having a matched edge in $H$ whose both endpoints are inner does not make progress in our computation.
To ensure that the dynamic matching oracle does not consider inner-inner edges in $G[U]$, instead of working with the original input graph, we work with a graph $G'$ obtained as follows. 
Given $G = (V, E)$, we make a bipartite graph $B = (L, R, E')$ where $R$ and $L$ are copies of $V$ and there exists an edge $\{x, y\}$ with $x \in L$ and $y \in R$ iff $\{x, y\} \in E$.
Then, we maintain the dynamic matching oracle on $B$ and not on $G$. 
When our framework samples a set of inner vertices $I$ and a set of outer vertices $O$, it invokes $G'[L \cap O, R \cap I]$.
Full details of this idea are presented in \cref{sec:framework-dynamic}.

In summary, our improvement in the static setting is obtained by refining the analysis of \cite{FMU22} and a new approach to partitioning the edges of $H$ into $1/\eps$ classes, each of which admits a more efficient simulation. In the dynamic setting, we propose a new vertex sampling paradigm that allows us to implement the framework with a much weaker oracle.
\section{Preliminaries}\label{sec:prelim}
We first introduce all the terminology, definitions, and notations. We also recall some well-known facts about blossoms.

Let $G$ be an undirected simple graph and $\eps \in (0, \frac{1}{4}]$ be the approximation parameter.
Without loss of generality, we assume that $\eps^{-1}$ is a power of $2$.
Denote by $V(G)$ and $E(G)$, respectively, the vertex and edge sets of $G$.
Let $n$ be the number of vertices in $G$ and $m$ be the number of edges in $G$.
An undirected edge between two vertices $u$ and $v$ is denoted by $\{u, v\}$.
Let $\mu(G)$ stand for the maximum matching size in $G$.
An \emph{$(\alp, \beta)$-approximate maximum matching} is a matching of size at least $\mu(G) / \alp - \beta$.
An \emph{$\alp$-approximate maximum matching} is a matching of size at least $\mu(G) / \alp$.
Throughout the paper, if not stated otherwise, all the notations implicitly refer to a currently given matching $M$, which we aim to improve.

\subsection{Alternating paths}

\begin{definition}[An unmatched edge and a free vertex] We say that an edge $\{u, v\}$ is \emph{matched} iff $\{u, v\} \in M$, and \emph{unmatched} otherwise.
We call a vertex $v$ \emph{free} if it has no incident matched edge, i.e., if $\{u, v\}$ are unmatched for all edges $\{u, v\}$.
Unless stated otherwise, $\alp, \beta, \gamma$ are used to denote free vertices. 
\end{definition}

\begin{definition}[Alternating and augmenting paths] An \emph{alternating path} is a simple path that consists of a sequence of alternately matched and unmatched edges. The \emph{length} of an alternating path is the number of edges in the path. An \emph{augmenting path} is an alternating path whose two endpoints are both free vertices.
\end{definition}

\subsection{Alternating trees and blossoms}

\begin{definition}[Alternating trees, inner vertices, and outer vertices]
\label{def:alternating-tree}
A subgraph of $G$ is an \emph{alternating tree} if it is a rooted tree where the root is a free vertex and every root-to-leaf path is an even-length alternating path.
An \emph{inner vertex} of an alternating tree is a non-root vertex $v$ such that the path from the root to $v$ is of odd length.
All other vertices are \emph{outer vertices}.
In particular, the root vertex is an outer vertex.
\end{definition}

Note that every non-root vertex in an alternating tree is matched.

\begin{definition}[Blossoms and trivial blossoms]
\label{def:blossom}
A blossom is identified with a vertex set $B$ and an edge set $E_B$ on $B$.
If $v \in V(G)$, then $B = \{v\}$ is a \emph{trivial blossom} with $E_B = \emptyset$.
Suppose there is an odd-length sequence of vertex-disjoint blossoms $A_0, A_1, \dots , A_k$ with associated edge sets $E_{A_0}, E_{A_1}, \dots, E_{A_k}$.
If $\{A_i\}$ are connected in a cycle by edges $e_0, e_1, \dots , e_k$, where $e_i \in A_i \times A_{i+1} (\mbox{modulo } k+1)$ and $e_1, e_3, \dots, e_{k-1}$ are matched, then $B = \bigcup_{i} A_i$ is also a blossom associated with edge set $E_B = \bigcup_i E_{A_i} \cup \{e_0, e_1, \dots , e_k\}$.
\end{definition}

Consider a blossom $B$.
A short proof by induction shows that $|B|$ is odd.
In addition, $M \cap E_B$ matches all vertices except one.
This vertex, which is left unmatched in $M \cap E_B$, is called the \emph{base} of $B$.
Note that $E(B) = E(G) \cap (B \times B)$ may contain many edges outside of $E_B$.
Blossoms exhibit the following property.

\begin{lemma}[\cite{DP14}]\label{lem:even-path} Let $B$ be a blossom. There is an even-length alternating path in $E_B$ from the base of $B$ to any other vertex in $B$. \end{lemma}

\begin{definition}[Blossom contraction] Let $B$ be a blossom. We define the contracted graph $G / B$ as the undirected simple graph obtained from $G$ by contracting all vertices in $B$ into a vertex, denoted by $B$.
\end{definition}

The following lemma is proven in \cite[Theorem 4.13]{edmonds1965paths}.

\begin{lemma}[\cite{edmonds1965paths}]\label{lem:contraction} Let $T$ be an alternating tree of a graph $G$ and $e \in E(G)$ be an edge connecting two outer vertices of $T$. Then, $T \cup \{e\}$ contains a unique blossom $B$. The graph $T / B$ is an alternating tree of $G / B$. It contains $B$ as an outer vertex. Its other inner and outer vertices are those of $T$ which are not in $B$.
\end{lemma}

Consider a set $\Omg$ of blossoms.
We say $\Omg$ is \emph{laminar} if the blossoms in $\Omg$ form a laminar set family.
Assume that $\Omg$ is laminar.
A blossom in $\Omg$ is called a \emph{root blossom} if it is not contained in any other blossom in $\Omg$.
Denote by $G / \Omg$ the undirected simple graph obtained from $G$ by contracting each root blossom of $\Omg$.
For each vertex in $\bigcup_{B \in \Omg} B$, we denote by $\Omg(v)$ the unique root blossom containing $v$.
If $\Omg$ contains all vertices of $G$, we denote by $M / \Omg$ the set of edges $\{ \{\Omg(u), \Omg(v)\} \mid \{u, v\} \in M \mbox{ and } \Omg(u) \neq \Omg(v) \}$ on the graph $G / \Omg$. It is known that $M / \Omg$ is a matching of $G / \Omg$ \cite{DP14}.

In our algorithm, we maintain \emph{regular} sets of blossoms, which are sets of blossoms whose contraction would transform the graph into an alternating tree satisfying certain properties.

\begin{definition}[Regular set of blossoms]
A regular set of blossoms of $G$ is a set $\Omg$ of blossoms satisfying the following:
\begin{itemize}
    \item[(C1)] $\Omg$ is a laminar set of blossoms of $G$. It contains the set of all trivial blossoms in $G$. 
    If a blossom $B \in \Omg$ is defined to be the cycle formed by $A_0, \dots, A_k$, then $A_0, \dots, A_k \in \Omg$.
    \item[(C2)] $G / \Omg$ is an alternating tree with respect to the matching $M / \Omg$. Its root is $\Omg(\alp)$ and each of its inner vertex is a trivial blossom (whereas each outer vertex may be a non-trivial blossom).
\end{itemize}    
\end{definition}

\subsection{Representation of edges and paths}
Each undirected edge $\{u, v\}$ is represented by two directed \emph{arcs} $(u, v)$ and $(v, u)$.
Let $(u, v)$ be an arc.
We say $(u, v)$ is \emph{matched} if $\{u, v\}$ is a matched edge;
otherwise, $(u, v)$ is \emph{unmatched}.
The vertex $u$ and $v$ are called, respectively, \emph{tail} and \emph{head} of $(u, v)$.
We denote by $\cev{(u, v)} = (v, u)$ the reverse of $(u, v)$.

Let $P = (u_1, v_1, \dots, \allowbreak u_k, v_k)$ be an alternating path, where $u_i$ and $v_i$ are vertices, $(u_i, v_i)$ are matched arcs, and $(v_i, u_{i+1})$ are unmatched ones.
Let $a_i = (u_i, v_i)$.
We often use $(a_1, a_2, \dots, a_k)$ to refer to $P$, i.e., we omit specifying unmatched arcs.
Nevertheless, it is guaranteed that the input graph contains the unmatched arcs $(v_i, u_{i+1})$, for each $1 \leq i < k$.
If $P$ is an alternating path that starts and/or ends with unmatched arcs, e.g., $P = (x, u_1, v_1, \dots, u_k, v_k, y)$ where $(x, u_1)$ and $(v_k, y)$ are unmatched while $a_i = (u_i, v_i)$ for $i = 1 \dots k$ are matched arcs, we use $(x, a_1, . . . , a_k, y)$ to refer to $P$. 

\subsection{Models of computation}

\paragraph{Massively Parallel Computation (MPC).}
The Massively Parallel Computation (MPC) model has become a standard for parallel computing, introduced in a series of papers~\cite{dean2008mapreduce,karloff2010model,goodrich2011sorting}. 
It is a theoretical abstraction of popular large-scale frameworks such as MapReduce, Flume, Hadoop, and Spark.
An MPC instance consists of $M$ machines whose communication topology is a clique. Each machine is characterized by its local memory of size $S$.
An MPC computation proceeds in synchronous rounds.
The input data is arbitrarily partitioned across the $M$ machines while ensuring that data sent to a machine is no larger than $S$.
During a round, each machine first performs computation locally. 
At the end of a round, the machines simultaneously exchange messages with the constraint that the total size of messages sent and received by a machine is at most $S$.

\paragraph{\CONGEST.}
Given a graph $G = (V, E)$, \CONGEST is a distributed model with $|V|$ machines with the topology between them being $E$.
The computation in this model proceeds in synchronous rounds. In each round, the machines perform computation independently; after that, each machine can send $O(\log n)$ bits of information along each edge.
Different information can be sent across different edges adjacent to the same machine.
The machines can perform arbitrary computations and use large amounts of space.

\paragraph{Semi-streaming model.}
In the semi-streaming model~\cite{FeigenbaumKMSZ05}, we assume the algorithm has no random access to the input graph.
The set of edges is represented as a stream.
In this stream, each edge is presented exactly once, and each time the stream is read, edges may appear in an arbitrary order.
The stream can only be read as a whole and reading the whole stream once is called a \emph{pass} (over the input).
The main computational restriction of the model is that the algorithm can only use $O(n \poly \log n)$ words of space, which is not enough to store the entire graph if the graph is sufficiently dense.

\paragraph{Dynamic.}
In a fully dynamic setting, we assume that edges in a graph are inserted and deleted. In an incremental (decremental) only setting, edges are only inserted (deleted).
Our algorithm is supposed to maintain a certain structure after each edge update. 
For example, after each update, it should be able to report a $1+\eps$ approximation of MCM.
In this setting, the goal is to reduce the time the algorithm needs to update the structure after an update.

\section{Review of the semi-streaming algorithm in \cite{MMSS25}} \label{sec:simp}

\subsection{Basic notation for the algorithm}

\paragraph{Vertex structures.}

In \cite{MMSS25}'s algorithm, each free vertex $\alp$ maintains a \emph{structure}, defined as follows. (See \cref{fig:structure} for an example)

\begin{definition}[The structure of a free vertex, \cite{MMSS25}] \label{def:structure}
The structure of a free vertex $\alp$, denoted by $\cS_\alp$, is a tuple $(G_\alp, \Omg_\alp, w'_\alp)$, where 
\begin{itemize}
    \item $G_\alp$ is a subgraph of $G$,
    \item $\Omg_\alp$ is a regular set of blossoms of $G_\alp$, and 
    \item $w'_\alp$ is either $\emptyset$ or an outer vertex of the alternating tree $G_\alp / \Omg_\alp$.
\end{itemize} 
Each structure $\cS_\alp$ satisfies the following properties.
\begin{enumerate}
    \item \textbf{Disjointness:} For any free vertex $\beta \neq \alp$, $G_\alp$ is vertex-disjoint from $G_\beta$.
    \item \textbf{Tree representation:} The subgraph $G_\alp$ contains a set of arcs satisfying the following: If $G_\alp$ contains an arc $(u, v)$ with $\Omg_\alp(u) \neq \Omg_\alp(v)$, then $\Omg_\alp(u)$ is the parent of $\Omg_\alp(v)$ in the alternating tree $G_\alp / \Omg_\alp$.
\end{enumerate}
\end{definition}

\noindent We denote the alternating tree $G_\alp / \Omg_\alp$ by $T'_\alp$.

\begin{definition}[The working vertex and active path of a structure, \cite{MMSS25}]
The \emph{working vertex} of $\cS_\alp$ is defined as the vertex $w'_\alp$, which can be $\emptyset$.
If $w'_\alp \neq \emptyset$, we define the \emph{active path} of $\cS_\alp$ as the unique path on $T'_\alp$ from the root $\Omg_\alp(\alp)$ to $w'_\alp$.
Otherwise, the active path is defined as $\emptyset$.
\end{definition}

\begin{definition}[Active vertices, arcs, and structures, \cite{MMSS25}]
A vertex or arc of $T'_\alp$ is said to be \emph{active} if and only if it is on the active path. We say $\cS_\alp$ is active if $w'_\alp \neq \emptyset$.    
\end{definition}

\begin{figure}[h]
\centering
    \begin{subfigure}[h]{0.27\linewidth}\label{fig:structure-graph}
        % \centering
        \includegraphics[width=\textwidth]{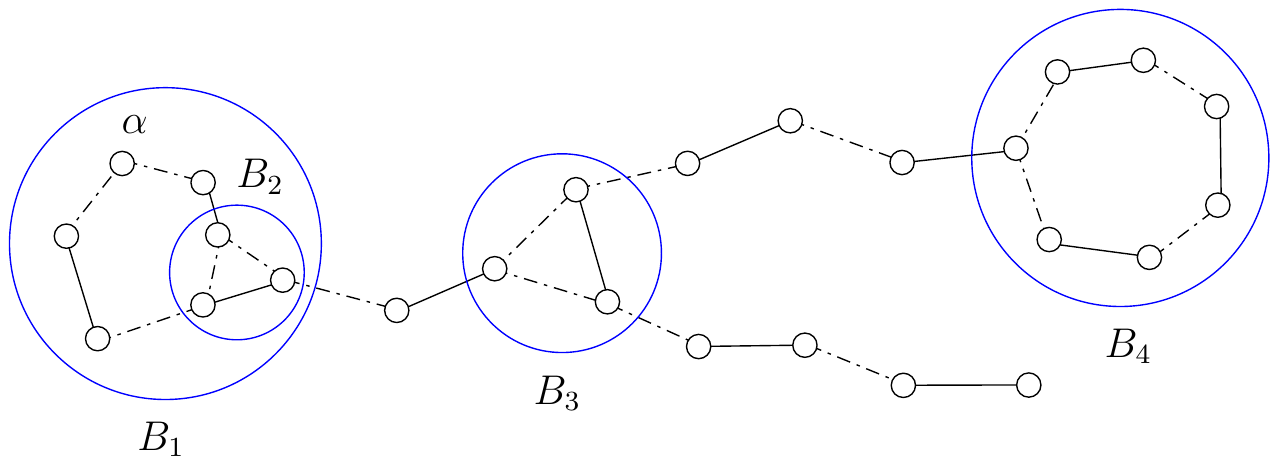}
        \caption{The graph $G_\alp$.}
        \label{fig:structure-a}
    \end{subfigure}
    \quad \quad
    % \vskip\baselineskip
     \begin{subfigure}[h]{0.22\linewidth}
        % \centering
        \includegraphics[width=\textwidth]{Sketches/structure-contracted.png}
        \caption{The contracted graph $G_\alp/\Omg_\alp$.}
        \label{fig:structure-b}
    \end{subfigure}
    \caption{Example of a structure $S_\alp$, where $\alp$ is a free vertex. Dashed and solid edges denote the unmatched and matched edges, respectively.
    \cref{fig:structure-a} shows the graph $G_\alp$.
    The set $\Omg_\alp$ contains all trivial blossoms in $G_\alp$ and the non-trivial blossoms $\{B_1, B_2\}$.  \cref{fig:structure-b} shows the corresponding contracted graph $G_\alp/\Omg_\alp$. The encircled vertices correspond to the non-trivial blossom $B_1$. The vertex $w'_\alp$ is the working vertex and the highlighted path, from $\alp$ to $w'_\alp$, is the active path.}
    \label{fig:structure}
\end{figure}

Let $F$ be the set of free vertices.
Throughout the execution, we maintain a set $\Omg$ of blossoms, which consists of all blossoms in $\bigcup_{\alp \in F} \Omg_\alp$ and all trivial blossoms.
Note that $\Omg$ is a laminar set of blossoms.
We denote by $G'$ the contracted graph $G / \Omg$.
The vertices of $G'$ are classified into three sets:
(1) the set of inner vertices, which contains all inner vertices in $\bigcup_{\alp \in F} V(T'_\alp)$;
(2) the set of outer vertices, which contains all outer vertices in $\bigcup_{\alp \in F} V(T'_\alp)$;
(3) the set of \emph{unvisited vertices}, which are the vertices not in any structure.

Similarly, we say a vertex in $G$ is \emph{unvisited} if it is not in any structure.
An arc $(u, v) \in G$ is a \emph{blossom arc} if $\Omg(u) = \Omg(v)$; otherwise, $(u, v)$ is a \emph{non-blossom arc}.
An \emph{unvisited arc} is an arc $(u, v) \in E(G)$ such that $u$ and $v$ are unvisited vertices.

\paragraph{Labels.}
The algorithm stores the set of all \emph{matched arcs} throughout its execution.
Each matched arc is associated with a \emph{label}, defined as follows.

\begin{definition}[The label of a matched arc, \cite{MMSS25}]
Each matched arc $a^* \in G$ is assigned a label $\ell(a^*)$ such that $1 \leq \ell(a^*) \leq \lmax + 1$, where $\lmax$ is defined as $3 / \eps$.
\end{definition}

Each matched arc $a' \in G'$ corresponds to a unique non-blossom matched arc $a \in G$;
for ease of presentation, we denote by $\ell(a')$ the label of $a$.
Their algorithm also maintains an invariant on the monotonicity of labels along alternating paths from the root.

\subsection{Overview of the algorithm}
\label{sec:algo-statement}

\cref{alg:outline} gives a high-level algorithm description of \cite{MMSS25}. 
Without loss of generality, we assume that $\frac{1}{\eps}$ is a power of $2$.

\begin{algorithm}
\begin{algorithmic}[1]
\medskip 
\Statex \textbf{Input:} a graph $G$ and the approximation parameter $\eps$
\Statex \textbf{Output:} a $(1+\eps)$-approximate maximum matching
\medskip
\Statex \hrule 

\State compute a 2-approximate maximum matching $M$ \label{line:2-approx-matching}
\For{scale $h = \frac{1}{2}, \frac{1}{4}, \frac{1}{8}, \dots, \frac{\epsilon^2}{64}$\label{line:scale-h}}
    \For {phases $t = 1, 2, \dots, \frac{144}{h\eps}$ \label{line:call-phase-given-h}} 
        \State $\calP \gets \algPhase(G, M, \eps, h)$ \Comment{Nothing stored from the previous phase.}
        \State restore all vertices removed in the execution of \algPhase \label{line:restore}
        \State augment the current matching $M$ using the vertex-disjoint augmenting paths in $\calP$ \label{line:augment}
    \EndFor
\EndFor
\State \Return $M$
\end{algorithmic}
\caption{A high-level algorithm description, \cite[Algorithm 1]{MMSS25}}
\label{alg:outline}
\end{algorithm}

In each phase, the procedure $\algPhase$ is invoked to find a set $\calP$ of vertex-disjoint augmenting paths.
In $\algPhase$, we may \emph{hypothetically} remove some vertices from $G$.
After $\algPhase$, \cref*{line:restore} restores all removed vertices to $G$.
Then, \cref*{line:augment} augments the current matching using the set $\calP$ of vertex-disjoint augmenting paths, which increase the size of $M$ by $|\calP|$.

\begin{algorithm}
\begin{algorithmic}[1]
\medskip 
\Statex \textbf{Input:} a graph $G$, the current matching $M$, the parameter $\eps$, and the current scale $h$
\Statex \textbf{Output:} a set $\calP$ of \emph{disjoint} $M$-augmenting paths
\medskip
\Statex \hrule 

\State $\calP \gets \emptyset$ \label{line:init-empty}
\State $\ell(a) \gets \lmax + 1$ for each arc $a \in M$ \label{line:init-label}
\State for each free vertex $\alp$, initialize its structure $\cS_\alp$ \label{line:init-structure}
\State compute parameters $\limit_h = \frac{6}{h} + 1$ and $\taumax(h) = \frac{72}{h\eps}$ \label{line:compute-par}
\For{$\bundle$s $\tau = 1, 2, \dots, \taumax(h)$} \label{line:for-pass}
    \For {each free vertex $\alp$} \label{line:init-for}
        \State if $\cS_\alp$ has at least $\limit_h$ vertices in $G$, mark $\cS_\alp$ as ``on hold'' \label{line:hold}
        \State if $\cS_\alp$ has less than $\limit_h$ vertices in $G$, mark $\cS_\alp$ as ``not on hold'' \label{line:not-hold}
        \State mark $\cS_\alp$ as ``not modified'' and ``not extended'' \label{line:not-modified}
    \EndFor
    \State $\algExtend$ (\cref{alg:extend}) \label{line:extend}
    \State $\algCheck$ \label{line:check}
    \State $\algBacktrack$ \label{line:backtrack}
\EndFor
\State \Return
\end{algorithmic}
\caption{$\algPhase$: the execution of a single phase, \cite[Algorithm 2]{MMSS25}
}
\label{alg:phase}
\end{algorithm}

\subsection{A phase (\algPhase)} 
\label{sec:phase}
In each phase, the algorithm executes DFS explorations from all free vertices in parallel; see \cref{alg:phase} for pseudocode.
\Cref*{line:init-empty,line:init-label,line:init-structure} initialize the set of paths $\calP$, the label of each arc, and the structure of each free vertex.
The structure of a free vertex $\alp$ is initialized to be an alternating tree of a single vertex $\alp$.
That is, $G_\alp$ and $\Omg_\alp$ are set to be
    a graph with a single vertex $\alp$ and
    a set containing a single trivial blossom $\{\alp\}$, respectively;
the working vertex $w'_\alp$ is initialized as the root of $T'_\alp$, that is, $\Omg_\alp(\alp)$.
The for-loop in \Cref*{line:for-pass} executes $\taumax(h)$ iterations, where each iteration is referred to as a \emph{\bundle}.
Each $\bundle$ consists of four parts:
\begin{enumerate}[(1)]
    \item \Cref*{line:init-for,line:hold,line:not-hold,line:not-modified} initialize the status of each structure in this $\bundle$. A structure is marked as \emph{on hold} if and only if it contains at least $\limit_h$ vertices. Each structure $\cS_\alp$ is marked as \emph{not modified} and \emph{not extended}. The purpose of this part is described in \cref{sec:marking}.
    \item $\algExtend$ makes a pass over the stream and attempts to extend each structure that is not on hold. Details of this procedure are given in \cref{sec:extend}.
    \item After $\algExtend$, $\algCheck$ is then invoked to identify blossoms and augmenting paths. 
    The procedure makes a pass over the stream, contracts some blossoms that contain the working vertex of a structure, and identifies pairs of structures that can be connected to form augmenting paths. Details of this procedure are given in \cref{sec:check}.
    \item The procedure $\algBacktrack$ examines each structure. If a structure is not on hold and fails to extend in this pass, $\algBacktrack$ backtracks the structure by removing one matched arc from its active path. Details of this procedure are given in \cref{sec:backtrack}.
\end{enumerate}

\subsection{Marking a structure on hold, modified, or extended} \label{sec:marking}
In the for-loop of \cref*{line:init-for}, we mark a structure $\cS_\alp$ \textit{on hold} if and only if it contains at least $\limit_h$ vertices.
See \cref*{line:hold,line:not-hold} of \cref{alg:phase}.

In the for-loop, we also mark each structure as \textit{not modified}.
Recall that each structure $\cS_\alp$ is represented by a tuple $(G_\alp, \Omg_\alp, w'_\alp)$;
in the execution of a $\bundle$, we mark a $\cS_\alp$ as modified whenever any of $G_\alp, \Omg_\alp$, or $w'_\alp$ is changed.
We also mark every structure as \emph{not extended}.
In the execution of $\algExtend$, we mark a structure as extended if it performs one of the basic operations presented in \cref{sec:basic-operations}.

\subsection{Basic operations on structures} \label{sec:basic-operations}
\cite{MMSS25} present three basic operations for modifying the structures. 
These operations are used to execute $\algExtend$ and $\algCheck$. 
Whenever one of these operations is applied, the structures involved are marked as modified and extended, except for a case: if a structure $\cS_\alp$ overtakes another structure $\cS_\beta$ (see \cref{sec:overtake}), only the overtaker ($\cS_\alp$) is marked as extended.

\subsubsection{Procedure $\algAugment(g, \calP)$}
\label{sec:augment}
\begin{itemize}
    \item \textbf{Invocation reason:} When the algorithm discovers an augmenting path in $G$.
    \item \textbf{Input:}
    
    - The set $\calP$.
    
    - An unmatched arc $g = (u, v)$, where $g \in E(G)$. The arc $g$ must satisfy the following property: $\Omg(u)$ and $\Omg(v)$ are outer vertices of two different structures.
\end{itemize}
Since $\Omg(u)$ is an outer vertex, $T'_\alp$ contains an even-length alternating path from the root $\Omg(\alp)$ to $\Omg(u)$.
Similarly, $T'_\beta$ contains an even-length alternating path from $\Omg(\beta)$ to $\Omg(v)$.
Since there is an unmatched arc $(\Omg(u), \Omg(v))$ in $G'$, the two paths can be concatenated to form an augmenting path $P'$ on $G'$.

By using \cref{lem:even-path}, we obtain an augmenting path $P$ on $G$ by replacing each blossom on $P'$ with an even-length alternating path.
$\algAugment$ adds $P$ to $\calP$ and removes $\cS_\alp$ and $\cS_\beta$.
That is, all vertices from $V(G_\alp) \cup V(G_\beta)$ are removed from $G$, and $\Omg$ is updated as $\Omg - (\Omg_\alp \cup \Omg_\beta)$.
The vertices remain removed until the end of $\algPhase$.
This guarantees that the paths in $\calP$ remain disjoint.
Recall that the algorithm adds these vertices back before the end of this phase, when \cref*{line:restore} of \cref{alg:outline} is executed.

\subsubsection{Procedure $\algContract(g)$}
\label{sec:contract}
\begin{itemize}
    \item \textbf{Invocation reason:} When a blossom in a structure is discovered.
    \item \textbf{Input:} 

    - An unmatched arc $g = (u, v)$, where $g \in E(G)$, such that $\Omg(u)$ and $\Omg(v)$ are distinct outer vertices in the same structure, denoted by $\cS_\alp$. In addition, $\Omg(u)$ is the working vertex of $\cS_\alp$
\end{itemize}
Let $g'$ denote the arc $(\Omg(u), \Omg(v))$.
By \cref{lem:contraction}, $T'_\alp \cup \{g'\}$ contains a unique blossom $B$.
The procedure contracts $B$ by adding $B$ to $\Omg_\alp$; hence, $T'_\alp$ is updated as $T'_\alp / B$ after this operation.
The arc $g$ is added to $G_\alp$.

By \cref{lem:contraction}, $T'_\alp$ remains an alternating tree after the contraction, and $B$ becomes an outer vertex of $T'_\alp$.
Next, the procedure sets the label of each matched arc in $E(B)$ to $0$.
(After this step, for each matched arc $a \in E(B)$, both $\ell(a)$ and $\ell(\cev{a})$ are $0$.)

Note that the working vertex of $\cS_\alp$, that is, $\Omg(u)$, is contracted into the blossom $B$.
The procedure then sets $B$ as the new working vertex of $\cS_\alp$.
Then, $\cS_\alp$ is marked as modified and extended.

\subsubsection{Procedure $\algOvertake(g, a, k)$}
\label{sec:overtake}
\begin{itemize}
    \item \textbf{Invocation reason:} When the active path of a structure $\cS_\alp$ can be extended through $g$ to \emph{overtake} the matched arc $a$ and reduce $\ell(a)$ to $k$.
    \item     \textbf{Input:}

    - An unmatched arc $g = (u, v) \in G$.

    - A non-blossom matched arc $a = (v, t) \in G$, which shares the endpoint $v$ with $g$.
    
    - A positive integer $k$.

    - The input must satisfy the following.
        \begin{itemize}
            \item[] (P1) $\Omg(u)$ is the working vertex of a structure, denoted by $\cS_\alp$.
            \item[] (P2) $\Omg(v) \neq \Omg(u)$, and $\Omg(v)$ is either an unvisited vertex or an inner vertex of a structure $\cS_\beta$, where $\cS_\beta$ can be $\cS_\alp$. In the case where $\Omg(v) \in \cS_\alp$, $\Omg(v)$ is not an ancestor of $\Omg(u)$.
            \item[] (P3) $k < \ell(a)$.
        \end{itemize}

\end{itemize}
For ease of notation, we denote $\Omg(u), \Omg(v)$, and $\Omg(t)$ by $u', v',$ and $t'$, respectively.
Since $v'$ is not an outer vertex, it is the trivial blossom $\{v\}$.
The procedure $\algOvertake$ performs a series of operations, detailed as follows.
Consider three cases, where in all of them we reduce the label of $a$ to $k$.

\paragraph{Case 1.} $a$ is not in any structure.
We include the arcs $g$ and $a$ to $G_\alp$.
The trivial blossoms $v'$ and $t'$ are added to $\Omg_\alp$.
The working vertex of $\cS_\alp$ is updated as $t'$, which is an outer vertex of $T'_\alp$.
Then, $\cS_\alp$ is marked as modified and extended.

\paragraph{Case 2.} $a$ is in a structure $\cS_\beta$.
By the definition of $g$, $v'$ is an inner vertex.
Thus, $v'$ is not the root of $T'_\beta$.
Let $p'$ be the parent of $v'$.
Two subcases are considered, where in both cases we re-assign the parent of $v'$ as $u'$ and make corresponding change in $G$.

\begin{itemize}
    \item[] \textbf{Case 2.1.} $\alp = \beta$. By (P2), $v'$ is not an ancestor of $u'$.
    We remove from $G_\alp$ all arcs $(p, v)$ such that $\Omg(p) = p'$.
    Then, $g$ is added to $G_\alp$.
    In $G'$, his operation corresponds to re-assigning the parent of $v'$ as $u'$.
    Then, we update the working vertex of $\cS_\alp$ as $t'$ and mark $\cS_\alp$ as modified.
    In addition, $\cS_\alp$ is marked as extended.

\item[] \textbf{Case 2.2.} $\alp \neq \beta$.
    See \cref{fig:overtake} for an example.
    Similar to Case 2.1, the objective of the overtaking operation is to re-assign the parent of $v'$ as $u'$ in $G'$.
    However, we need to handle several additional technical details in this case.
    The overtaking operation consists of the following steps.
    \begin{itemize}
        \item[] Step 1: Remove from $G_\beta$ all arcs $(p, v)$ such that $\Omg(p) = p'$; add the arc $(u, v)$ to $G_\alp$.
        \item[] Step 2: Move, from $G_\beta$ to $G_\alp$, all vertices $x$ such that $\Omg(x)$ is in the subtree of $v'$
        \item[] Step 3: Move, from $G_\beta$ to $G_\alp$, all arcs $(x, y)$ where $x$ and $y$ are both moved in Step 2.
        \item[] Step 4: Move, from $\Omg_\beta$ to $\Omg_\alp$, all blossoms that contain a subset of vertices moved in Step 2.
        \item[] Step 5: If the working vertex of $\cS_\beta$ was under the subtree of $t'$ before Step 1, we set $w'_\alp$ as $w'_\beta$ and then update $w'_\beta$ as $\Omg(p)$. Otherwise, set $w'_\alp$ as $t'$.
    \end{itemize}
    After the overtaking operation, both $\cS_\alp$ and $\cS_\beta$ are marked as modified, and only $\cS_\beta$ is marked as extended.
\end{itemize}

\begin{figure}[h]
\centering
    \begin{subfigure}[h]{0.6\linewidth}
        % \centering
        \includegraphics[width=\textwidth]{Sketches/overtake1.png}
        \caption{Before \algOvertake.}
        \label{fig:overtake-a}
    \end{subfigure}
    % \hfill
    \vskip\baselineskip
     \begin{subfigure}[h]{0.6\linewidth}
        % \centering
        \includegraphics[width=\textwidth]{Sketches/overtake2.png}
        \caption{After \algOvertake.}
        \label{fig:overtake-b}
    \end{subfigure}
    \caption{Example of Case 2.2 of the procedure $\algOvertake$ where $g = (u, v)$ connects the two structures $S_\alp$ and $S_\beta$. Although in this example $\Omg(p) = \{p\}$ is a trivial blossom, it can be a non-trivial blossom in general.}
    \label{fig:overtake}
\end{figure}

\subsection{Procedure \algExtend} \label{sec:extend}
The goal of $\algExtend$ is to \textit{extend} each structure $\cS_\alp$, where $\cS_\alp$ is not on hold, by performing at most one of the $\algAugment$, $\algContract$, or $\algOvertake$ operations.

\begin{algorithm}
\begin{algorithmic}[1]
\medskip 
\Statex \textbf{Input:} a graph $G$, the parameter $\eps$, the current matching $M$, the structure $\cS_\alp$ of each free vertex $\alp$, the set of paths $\calP$
\medskip 
\Statex \hrule 

\For{each arc $g = (u, v) \in E(G)$ on the stream}
    \If{$u$ or $v$ was removed in this phase}
        \State continue with the next arc
    \EndIf
    
    \If{$\Omg(u) = \Omg(v)$, or $\Omg(u)$ is not the working vertex of any structure, or $g$ is matched}
        \State continue with the next arc
    \EndIf
    \If{$u$ belongs to a structure that is marked as modified or on hold}
        \State continue with the next arc
    \EndIf
    \If {$\Omg(v)$ is an outer vertex}
        \If {$\Omg(u)$ and $\Omg(v)$ are in the same structure}
            \State $\algContract(g)$
        \Else
            \State $\algAugment(g)$
        \EndIf
    \Else \Comment{$\Omg(v)$ is either unvisited or an inner vertex.}
        \State compute $\lab(u)$ \label{line:shortest-path}
        \State $a \gets$ the matched arc in $G$ whose tail is $v$
        \If {$\lab(u) + 1 < \ell(a)$}
            \State $\algOvertake(g, a, \lab(u) + 1)$
        \EndIf
    \EndIf    
\EndFor
\end{algorithmic}
\caption{The execution of $\algExtend$.} \label{alg:extend}
\end{algorithm}

The procedure works as follows.
(See \cref{alg:extend} for a pseudocode.)
The algorithm makes a pass over the stream to read each arc $g = (u, v)$ of $G$.
When an arc $g$ is read, it is mapped to an arc $g' = (\Omg(u), \Omg(v))$ of $G'$.
In $\algExtend$, we only consider non-blossom unmatched arcs whose tail is a working vertex.
Hence, if $\Omg(u) = \Omg(v)$, $\Omg(u)$ is not the working vertex of a structure, or $g$ is a matched arc, then we simply ignore $g$.
If one of $u$ or $v$ is removed, we also ignore $g$.

Let $\cS_\alp$ denote the structure whose working vertex is $\Omg(u)$, and let $\cS_\beta$ denote the structure containing $\Omg(v)$.
We ignore $g$ if $\cS_\alp$ is marked as on hold or extended.
Thus, we also ignore $g$ if $\cS_\alp$ is marked as extended.
This ensures that each structure only extends once in the execution of $\algExtend$.
(Recall that when a structure is overtaken, it is marked as modified but not extended; therefore, it may still extend in this execution.)

We examine whether $g$ can be used for extending $\cS_\alp$ as follows.
\begin{itemize}
    \item[] \textbf{Case 1:} $\Omg(v)$ is an outer vertex and $\cS_\alp = \cS_\beta$. In this case, $g'$ induces a blossom on $T'_\alp$. We invoke $\algContract$ on $g$ to contract this blossom.
    \item[] \textbf{Case 2:} $\Omg(v)$ is an outer vertex and $\cS_\alp \neq \cS_\beta$. In this case, the two structures can be connected to form an augmenting path. We invoke $\algAugment$ to compute this augmenting path and remove the two structures.
    \item[] \textbf{Case 3:} $\Omg(v)$ is either an inner vertex or an unvisited vertex.
    Note that $v$ cannot be a free vertex because, for each free vertex $\gamma$, it holds that $\Omg(\gamma)$ is an outer vertex.
    Therefore, $v$ is the tail of a matched arc $a$.
    We determine whether $\cS_\alp$ can overtake $a$ by computing a number $\lab(u) + 1$ and compare it with $\ell(a)$.
    The number $\lab(u)$ represents the last label in the active path, which is computed as follows:
    If $\Omg(u)$ is a free vertex, $\lab(u)$ is set to $0$;
    otherwise, $\lab(u)$ is the label of the matched arc in $G'$ whose head is $\Omg(u)$.
    If $\lab(u) + 1 < \ell(a)$, $\algOvertake$ is invoked to update the label of $a$ as $\lab(u) + 1$.
\end{itemize}

\subsection{Procedure \algCheck} \label{sec:check}
The procedure $\algCheck$ performs two steps to identify augmenting paths and blossoms:
\begin{itemize}
    \item[] Step 1: Repeatedly invoke $\algContract$ on an arc connecting two outer vertices of the same structure, where one of the outer vertices is the working vertex. This operation is repeated until no such arcs exist.
    \item[] Step 2: Repeatedly perform \algAugment on an arc $g$ connecting outer vertices of different structures, until no such arcs exist.
\end{itemize}
We omit the implementation details in the streaming model.

\subsection{Procedure \algBacktrack} \label{sec:backtrack}
For each structure $\cS_\alp$ that is not on hold and not modified, \algBacktrack performs the backtrack operation as follows.
If $w'_\alp$ is a non-root outer vertex of $T'_\alp$, we update the working vertex as the parent of the parent of $w'_\alp$, which is an outer vertex.
Otherwise, set the working vertex of $\cS_\alp$ as $\emptyset$, which makes $\cS_\alp$ inactive.

\subsection{Properties of the algorithm}
Let $\Delta_h = 36h/\eps$ for each scale $h$.
Let $\Delta = 2304 \eps^{3}$, which is the minimum value of $\Delta_h$ accross all scales.
\cite{MMSS25} proved the following.

\begin{lemma}[{\cite{MMSS25}}, upper bound on the the size of structures] \label{lem:struct-size}
At any point of a scale $h$, each structure contains at most $\Delta_h$ vertices in $G$.
\end{lemma}

\begin{theorem}[{\cite{MMSS25}}] \label{thm:MMSS}
\cref{alg:outline} outputs a $(1+\eps)$-approximate maximum matching.
\end{theorem}

% \subsection{Our modifications to the algorithm}
% The following modification is made to \cite{MMSS25}'s algorithm:
% In addition to the ``modified'' mark, we also mark some structures as ``extended''.
% The only difference of the two is that an overtaken structure is considered modified but not extended, and hence may still extend in the same \bundle.
% This change is to make the algorithm easier to simulate in other models.
% In \cref{sec:app-correctness}, we show that the algorithm (with some additional change due to the simulation) is still correct after the change.
% Most of the correctness carries over unchanged from \cite{MMSS25}.

\section{Boosting framework for graph oracle} \label{sec:framework-distr}

This section describes a boosting framework for computing a $(1+\eps)$-approximate matching.
The framework assumes oracle access to an algorithm for computing a constant approximate matching.

\begin{definition}[$\Amat$] \label{def:Amat}
Given a graph $H$, the algorithm $\Amat$ returns a $c$-approximate matching of $H$, where $c > 1$ is a constant.
\end{definition}

The boosting framework computes a $(1+\eps)$-approximate matching by simulating the semi-streaming algorithm.
In the simulation, it invokes $\Amat$ $O(c \log(1/\eps) / \eps^{7})$ times to simulate $\algExtend$ (and a few other procedures).
The framework also requires a few basic operations on graphs, e.g. exploring connected subgraphs of size $\peps$.
(These connected subgraphs correspond to the structures in the semi-streaming algorithm.)
We call these operations $\Aexp$.

In the following, we present a high-level description of this framework, omitting all model-specific details and focusing on how $\Amat$ is used for the simulation, as it is the main difference between our and \cite{FMU22}'s frameworks.
In \cref{sec:app-framework-distr}, we formally define $\Aexp$ and provide implementation details in \MPC and \CONGEST.
In \cref{sec:framework-dynamic,sec:dynamic}, we adapt this framework to solving the dynamic $(1+\eps)$-approximate matching problem.

The following theorem summarizes the framework.

\frameworkdistr*

\noindent \cref{thm:framework-1} improves on previous frameworks, developed in \cite{FMU22, MMSS25}, that require $\Omega(c \log^2 c /\eps^{52})$ and $\Omega(c \log^2 c /\eps^{39})$ invocation of $\Amat$, respectively.

\begin{remark}
    As \cite{FMU22}'s framework, our framework works even if $c$ is a non-decreasing function of $n$ and $m$. (E.g., the framework works even if $\Amat$ returns a $\log n$ approximation.) Furthermore, if the input graph $G$ has maximum degree $\leq D$ and arboricity $\leq L$, then $\Amat$ is always invoked on a graph with maximum degree $\leq \frac{2}{\eps^3} D$ and arboricity $\leq \frac{2}{\eps^3} L$.
\end{remark}

\subsection{Notations}
To distinguish between the vertices of $G$ and $G'$, we use $G$-vertex (resp. $G'$-vertex) when we refer to a vertex in $G$ (resp. a vertex in $G'$).
The terms \emph{$G$-arc} and \emph{$G'$-arc} are defined similarly.
The \emph{size} of a structure $\cS_\alp$, denoted by $|\cS_\alp|$, is the number of $G$-vertices in it.

For ease of presentation, we extend the notation for structures and labels as follows.
Recall that $\cS_\alp$ denotes the structure of a free vertex $\alp$.
For each $G$-vertex $v$ (resp. $G'$-vertex $v'$), we also let $\cS_v$ represent the structure containing $v$ (resp. $v'$).
Note that the structure containing a non-free vertex $v$ may change when $v$ is overtaken, while the structure containing a free vertex is fixed throughout a phase.
For each matched vertex $v' \in G'$, denote by $\ell(v')$ the label of the matched arc adjacent to $v'$;
For a free vertex $\alp' \in G'$, define its label $\ell(\alp')$ as $0$.
We define $\ell(v)$ for each $G$-vertex $v$ in a similar way.

\subsection{Overview of the framework}
The framework simulates the semi-streaming algorithm (\cref{alg:outline}).
Most steps of the algorithm can be simulated in a straightforward way.
The main challenge is to simulate the following three procedures: computing the initial matching, $\algCheck$, and $\algExtend$.

\paragraph{Computing the initial matching.}
We compute the initial matching as a 4-approximate matching, instead of a 2-approximation.
This does not affect the correctness -- the algorithm still outputs a $(1+\eps)$-approximation if we increase the number of phases by a constant factor.

The computation is done by iteratively calling $\Amat$ and removing all matched vertices until all removed vertices form a 4-approximation.
We show that $O(c)$ calls suffice.

\paragraph{Simulating $\algCheck$.}
Essentially, the simulation of $\algCheck$ and $\algExtend$ is to find $G'$-arcs on which the three basic operations ($\algContract$, $\algAugment$, and $\algOvertake$) can be performed.
The following notions characterize such arcs.

\begin{definition}[type 1 arc, type 2 arc, type 3 arc] \label{def:arc-type}
    Let $a' = (u', v')$ be an arc in $G'$.
    We say $a'$ is of \emph{type 1} if it connects two outer vertices of some structure $\cS$, and one of them is the working vertex of $\cS$;
    it is of type 2 if it connects outer vertices of two different structures;
    it is of type 3 if all of the following are satisfied:
    \begin{itemize}
        \item $u'$ is the working vertex of $\cS_{u'}$,
        \item $v'$ is an inner vertex,
        \item $\ell(v') > \ell(u') + 1$, and
        \item $\cS_{u'}$ is not on-hold,
    \end{itemize}
    
    We say an arc $(u, v) \in G$ is of type 1, 2, or 3 if its corresponding arc $(\Omg(u), \Omg(v)) \in G'$ is of type 1, 2, or 3, respectively.
\end{definition}

\noindent Note that types 1, 2, and 3 arcs are, respectively, the arcs on which $\algContract$, $\algAugment$, and $\algOvertake$ can be performed.

Recall that the procedure $\algCheck$ consists of two steps.
On a high level, the two steps are to exhaustively perform $\algContract$ (resp. \algAugment) on type 1 (resp. type 2) arcs in the graph.
Step 1 does not require invocations of $\Amat$; it can be done by examining the in-structure arcs (i.e. the arcs with both endpoints in the same structure).
Implementation of this step is simple but model-specific;
hence we only describe Step 2 in this section.

Step 2 simulates $\algAugment$, which removes structures that are connected by a type 2 arc.
Thus, if a structure is adjacent to two type 2 arcs, we can only perform $\algAugment$ on one of them.
Based on this property, performing $\algAugment$ can be phrased as a matching problem:
Suppose that the semi-streaming algorithm performs $\algAugment$ on a set of type 2 arcs $N'$.
Then, $N'$ must form a matching between structures.
(More precisely, each structure contains at most one endpoint of edges in $N'$.)
Hence, the simulation first obtains a graph $H'$ by contracting every structure and removing all non-type 2 arcs.
Then, it iteratively finds a matching in $H'$; if two structures are matched in this matching, $\algAugment$ is performed to remove both of them (and record the augmenting path in between).
By a single argument, it can be shown that each iteration reduces $\mu(H')$ by a constant factor; i.e. $\mu(H')$ drops exponentially throughout iterations.
However, since we only invoke $\Amat$ $\peps$ times, $\mu(H')$ never drops to $0$, meaning there must be some type 2 arcs where we fail to find.
We mark all remaining type 2 arcs as \emph{contaminated}, which represents that $\algAugment$ should have been performed on some of these arcs, but our simulation fails to find them.
We show that the set of contaminated arcs admits a small vertex cover, which in turn implies that they only intersect a small number of disjoint augmenting paths;
By running $O(c \ln(1 / \eps))$ iterations, it can be shown that the framework can still find a $(1+\eps)$-approximate matching even though it does not simulate the semi-streaming algorithm. 

\paragraph{Simulating $\algExtend$.}
$\algExtend$ performs \algContract, \algAugment, and \algOvertake on arcs of type 1, 2, and 3, respectively.
The first step of the simulation is to address type 3 arcs -- that is, the simulation procedure repeatedly finds type 3 arcs from the graph and performs \algOvertake on them, until type 3 arcs are almost exhausted.

This step is similar to the simulation of $\algCheck$, except for a few modifications to handle an additional technical difficulty.
Details are given as follows.
The simulation consists of several iterations.
In each iteration, we invoke $\Amat$ in a derived graph $H'$ to find a matching consisting of type 3 arcs.
Then, we modify the structures by performing $\algOvertake$ on each arc in the returned matching.

There is a key difference between $\algOvertake$ and $\algAugment$: The streaming algorithm does not remove the two structures involved in an $\algOvertake$ operation.
Therefore, when $\algOvertake$ is performed on an arc in one iteration, the two modified structures will still participate in the next iteration.
Due to this property, we cannot apply the same analysis used for $\algCheck$.
More precisely, we cannot use the same argument to show that $\mu(H')$ decreases by a constant factor in each iteration.
Therefore, it is now unclear whether $\ln(1/\eps)$ iterations are enough for the simulation.

To address this issue, we propose a slightly modified simulation and show that the modified version requires only $\ln(1/\eps) / \eps$ iterations.
The modifications are as follows.
We split the simulation into $\lmax$ stages, labeled $1, 2, \dots, \lmax$.
In a stage $s \in \{1, 2, \dots, \lmax\}$, the goal is to find type 3 arcs $(u', v')$ such that $\ell(u') = s$ (i.e. the overtaker has label $s$).
To this end, we construct the subgraph of $G'$ induced by the set of type 3 arcs $(u', v')$ with $\ell(u') = s$ in each stage $s$.
Denote this subgraph by $H'_s$.
We execute $\ln(1/\eps)$ iterations in each stage, where in each iteration we invoke $\Amat$ on $H'_s$ and perform \algOvertake on the arcs in the returned matching. 
It can be shown that $\mu(H'_s)$ decreases by a constant in each iteration, which in turn implies that $O(\ln(1/\eps))$ iterations are enough to decrease $\mu(H'_s)$ to a negligible number.
Since there are only $O(1/\eps)$ stages, in total $O(\ln(1/\eps) / \eps)$ invocations are required.

A technical detail is as follows.
As shown above, executing $O(\ln(1/\eps))$ iterations in a stage $s$ can only guarantee that $\mu(H'_s)$ decreases to a negligible number.
Therefore, at the end of stage $s$, $H'_s$ still contains a few type 3 arcs, on which we could have performed \algOvertake.
As before, we mark these arcs as contaminated.
It can be shown that, by running an appropriate number of iterations, all stages generate only a negligible number of contaminated arcs.

After completing all stages, there are no type 3 arcs (except the contaminated ones) in the graph.
We run our simulation of \algCheck to handle all arcs of type 1 or 2.
It can be shown that this final step does not create new type 3 arcs.
Consequently, after the simulation, the graph contains no arcs of type 1, 2, or 3 except the contaminated ones.

\subsection{Computation of the initial matching}
The first step of the algorithm is to compute a $4$-approximate matching $M$.

\begin{lemma} \label{lem:init}
A $4$-approximate matching can be computed with $2c$ calls to $\Amat$. \end{lemma}
\begin{proof}
The computation is done by finding matchings in $G$ iteratively.
We initialize $M$ as an empty matching.
In each iteration, we invoke $\Amat$ on the subgraph of $G$ induced by all unmatched vertices.
All matched edges returned by $\Amat$ are added to $M$.

Let $G_i$ denote the subgraph on which $\Amat$ is invoked in the $i$-th iteration.
Fix an iteration $i$.
We claim that $\mu(G_{i+1}) \leq (1 - \frac{1}{c}) \mu(G_i)$ for all $i$ except the last iteration.
To see this, consider two matching $M_i$ and $M^*_{i+1}$, where $M_i$ is the matching found by $\Amat$ in $i$-th iteraion and $M^*_{i+1}$ is the largest matching in $G_{i+1}$.
Since we remove all matched vertex in each iteration, $M_i$ and $M^*_{i+1}$ are disjoint.
This implies that $M_i \cup M^*_{i+1}$ is a matching in $G_i$ of size $|M_i| + |M^*_i| \geq \frac{1}{c}\mu(G_i) + \mu(G_{i+1})$.
It follows that $\frac{1}{c}\mu(G_i) + \mu(G_{i+1}) \leq \mu(G_i)$, or equivalently
\begin{equation}\label{eqn:drop}
    \mu(G_{i+1}) \leq (1 - \frac{1}{c}) \mu(G_i).
\end{equation}
\noindent Hence, after $2c$ iterations, the matching size in $G_\tau$ is at most

\begin{equation} \label{eqn:approx-mat}
    (1 - \frac{1}{c})^{2c} \mu(G) \leq e^{-\frac{1}{c} \cdot 2c} \mu(G) = e^{-2} \mu(G) \leq \frac{1}{4} \mu(G).
\end{equation}

Let $M_{2c}$ be the maximum matching of $G_{2c}$.
Note that $M \cup M_{2c}$ is an inclusion-wise maximal matching of $G$.
Hence, $|M| + |M_{2c}| \geq \frac{1}{2} \mu(G)$.
Combining the equation with \cref{eqn:approx-mat}, we have $|M| \geq \frac{1}{4} \mu(G)$.
This completes the proof.
\end{proof}

\subsection{Simulation of $\algCheck$} \label{sec:sim-check-1}
Recall that $\algCheck$ consists of two steps:

\begin{itemize}
    \item[] Step 1: Repeatedly invoke $\algContract$ on an arc connecting two outer vertices of the same structure, where one of the outer vertices is the working vertex.
    \item[] Step 2: For each arc $g$ connecting outer vertices of different structures, invoke $\algAugment$ with $g$.
\end{itemize}

\begin{algorithm}
\begin{algorithmic}[1]
\medskip 
\Statex \textbf{Input:} a graph $G$, the parameter $\eps$, the current matching $M$, the structure $\cS_\alp$ of each free vertex $\alp$, the set of paths $\calP$
\medskip
\Statex \hrule 
\State simulate Step 1 of $\algCheck$ \Comment{After this line, the graph contains no type 1 arc}
\State construct $H'$ as defined in \cref{def:Hp}
\For {$22 c \ln(1/\eps)$ iterations}
    \State find a matching $M'$ by invoking $\Amat$ on $H'$
    \For {each arc $(u', v')$ in $M'$}
        \State perform $\algAugment$ on $(u', v')$
    \EndFor
\EndFor
\For {each type 2 arc $e \in G$}
    \State mark $e$ as contaminated \Comment{This step is only for the analysis}
\EndFor
\end{algorithmic}
\caption{Simulating $\algCheck$ using $\Amat$.} \label{alg:sim-check-1}
\end{algorithm}

\noindent The purpose of $\algCheck$ is to perform $\algContract$ and $\algAugment$ until the graph contains no arcs of type 1 or 2.
A detailed description of our simulation is as follows.
(See \cref{alg:sim-check-1}.)
First, we simulate Step 1 of $\algCheck$.
The implementation of this step is model-specific and hence deferred to \cref{sec:framework-dynamic,sec:app-framework-distr}.
In the following, we assume that Step 1 is simulated exactly; that is, $G$ contains no type 1 arc.

To simulate Step 2, we construct the following graph $H'$.
(See \cref{fig:sim-augment} for an example.)

\begin{definition} \label{def:Hp}
Given $G$, $H'$ is constructed as the graph such that:
\begin{itemize}
    \item The vertex set of $H'$ is the set of structures in $G$.
    \item $H'$ contains an arc $(\cS_1, \cS_2)$ if and only if there is an arc in $G$ that connect outer vertices of the two structures.
\end{itemize}
Note that every arc in $H'$ is of type 2.
\end{definition}

\begin{figure}[h]
\centering
    \begin{subfigure}[h]{0.52\linewidth}
        % \centering
        \includegraphics[width=\textwidth]{Sketches/sim-augment1.png}
        \caption{The graph $G$.}
        \label{fig:sim-augment-a}
    \end{subfigure}
    \quad \quad \quad
    % \vskip\baselineskip
     \begin{subfigure}[h]{0.26\linewidth}
        % \centering
        \includegraphics[width=\textwidth]{Sketches/sim-augment2.png}
        \caption{The corresponding graph $H'$.}
        \label{fig:sim-augment-b}
    \end{subfigure}
    \caption{An example of a graph $G$ and the corresponding graph $H'$.
    \cref{fig:sim-augment-a} shows $G$, which contains three structures $\cS_\alp, \cS_\beta, \cS_\gamma$ and some edges in between.
    \cref{fig:sim-augment-b} gives $H'$, where each structure is contracted into a vertex represented by a triangle.
    The edge set of $H'$ contains pairs of structures that are connected by a type 2 arc.
    Note that $H'$ does not contain the edge $(\cS_\alp, \cS_\gamma)$ because the edge $(u, v) \in G$ does not connect two outer vertices of $\cS_\alp, \cS_\beta$.}
    \label{fig:sim-augment}
\end{figure}

\noindent The simulation is iterative.
In each iteration, we invoke $\Amat$ on $H'$ to find a matching $M'$ consisting of type 2 arcs.
For each arc in $\Amat$, we perform $\algAugment$ on the two matched structures.
This removes all matched structures from $H'$ (because $\algAugment$ removes the two structures involved).
The above procedure is repeated for $22 c \ln(c/\eps)$ iterations.

At the end of the simulation, $H'$ may still contain some type 2 arcs that are never found by $\Amat$.
After the last iteration, we mark all type 2 arcs in $G$ as \emph{contaminated}, representing that these arcs are missed by our simulation.
Let $E_C$ denote the set of contaminated arcs marked in the simulation.
We remark that the contaminated arcs are identified solely for the analysis.
That is, our simulation works even if we do not mark the remaining type 2 arcs as contaminated.
(In particular, our implementation in the dynamic model does not mark contaminated arcs. See \cref{sec:framework-dynamic}.)

\paragraph{Correctness.}
Recall that the goal is to perform $\algContract$ and $\algAugment$ until the graph contains no arcs of type 1 or 2.
Our simulation is inexact only because it does not find all type 2 arcs.
In \cref{alg:sim-check-1}, these arcs are marked as contaminated.
Each contaminated arc represents a potential augmentation that is missed by our simulation.
To prove the correctness, we show that we can still find a $(1+\eps)$-approximate matching even if these augmentations are missed.
Details of this proof are given in \cref{sec:app-correctness}.
In the following, we present a key property used in the proof, showing that the contaminated arcs only intersect a negligible number of augmenting paths.
In other words, the number of augmenting paths missed by our simulation is negligible.

We first show that $\mu(H')$ is dropping exponentially in our simulation.

\begin{lemma}
\label{lem:drop}
After the last iteration of \cref{alg:sim-check-1}, it holds that $\mu(H') \leq \eps^{20} |M|$.
\end{lemma}
\begin{proof}
We first show that $\mu(H') \leq 3|M|$ before the start of the first iteration. 
Recall that each $\algAugment$ finds one $M$-augmenting path between two structures.
Hence, if $H'$ contains a matching of size $x$, then $\algAugment$ can be used to find $x$ disjoint augmenting paths in $G$.
Since we start from a 4-approximate matching, $G$ contains at most $3|M|$ augmenting paths at any point of the algorithm.
Hence, $\mu(H')$ is at most $3|M|$ before the first iteration.

In each iteration, we find a $c$-approximate matching in $H'$ and remove all matched structures.
Similar to the proof of \cref{eqn:drop} in \cref{lem:init}, we can show that $\mu(H')$ is decreased by a factor of $(1 - 1/c)$ after each iteration.
Therefore, after $22 c \ln(1 / \eps)$ iterations, $\mu(H')$ decreases by a factor of

\[(1 - 1/c)^{22 c \ln(1 / \eps)} \leq e^{-22 \ln (1 / \eps)} = \eps^{22} \leq \frac{\eps^{20}}{3}, \]

\noindent where the last inequality holds for $\eps \leq 0.5$.
This completes the proof.
\end{proof}

Recall that our simulation of $\algCheck$ is not exact because it does not perform $\algAugment$ on the set $E_C$ of contaminated arcs.
In the following, we argue that $E_C$ only intersects a small number of augmenting paths, implying that our simulation only misses a negligible number of augmentations.

\begin{lemma} \label{lem:vc} Let $E_C$ be the set of contaminated arcs marked in \cref{alg:sim-check-1}. Then, $E_C$ only intersects $2 \eps^{17} |M|$ vertex-disjoint augmenting paths in $G$.
\end{lemma}
\begin{proof}
Let $M_{H'}$ be the maximum matching of $H'$ at the end of the last iteration.
By \cref{lem:drop}, $|M_{H'}| \leq \eps^{20} |M|$.
Since $M_{H'}$ is maximal, $V(M_{H'})$ form a vertex cover of $H'$, whose size is $2|M_{H'}| \leq 2 \eps^{20} |M|$.
Denote this vertex cover by $S'$.

Recall that each vertex in $H'$ represents a structure in $G$.
Let $S$ be the set of $G$-vertices that are contained in a structure in $S'$.
It is not hard to see that $S$ covers all edges in $E_C$.
In addition, $S$ contains at most $|S'| \cdot \Delta \leq 2 \eps^{17} |M|$ vertices.
Therefore, the size of a maximum matching in $E_C$ is at most $2 \eps^{17} |M|$.
Consequently, $E_C$ can intersect at most $2 \eps^{17} |M|$ vertex-disjoint augmenting path.
\end{proof}

Recall that each phase has $\taumax = 1/\eps^3$ $\bundle$s.
By \cref{lem:vc}, the total number of augmenting paths missed in a phase is at most $\taumax \cdot \eps^{20} |M| = \eps^{17} |M|$, which is negligible compared to the size of $M$.
\cref{sec:app-correctness} gives a formal correctness proof that uses this property.

% \paragraph{Remark on correctness.} Each time we find a matching in $H'$, it's like we make a pass over those edges and perform $\algAugment$ on them. Thus, there exists an order of edges such that the behavior of the semi-streaming algorithm is exactly the same as our simulation. 
% After all iterations we remove some edges from the graph; this does not affect the correctness of the algorithm. (This claim may need some explanation.) 

% \btodo{One drawback of this framework is that $\Amat$ is always invoked on a new graph in which each vertex represents a whole structure. In \CONGEST, this implies that each round on $H'$ actually requires $\Delta$ rounds on the original graph.
% I am thinking whether this problem can be overcome by a more careful simulation.}

\subsection{Simulation of $\algExtend$} \label{sec:sim-extend-1}

\paragraph{Simulation procedure.} Recall that $\algExtend$ makes a pass over the arcs, and performs $\algContract$, $\algAugment$, or $\algOvertake$ on it whenever applicable. A detailed description of our simulation is given below (see \cref{alg:sim-extend-1} for pseudo code).
We first focus on the simulation of $\algOvertake$.
The simulation consists of $\lmax$ stages, labeled $1, 2, \dots, \lmax$.
In each stage $s$, the goal is to simulate $\algOvertake$ on arcs $(u', v') \in G'$ with $\ell(u') = s$.
These arcs are called \emph{$s$-feasible arcs}, formally defined as follows.

\begin{definition}[$s$-feasible arc] \label{def:s-feasible}
For integer $s \in [0, \lmax]$, define an $s$-feasible arc as a type 3 arc $(u', v')$ with $\ell(u') = s$.
\end{definition}

In each stage, we work on a bipartite subgraph of $H'_s$ consisting of $s$-feasible arcs.

\begin{definition}[Bipartite graph $H'_s$] \label{def:Hps}
$H'_s$ is the bipartite subgraph of $G'$ constructed as follows.
The left part consists of all outer vertices $u'$ satisfying the following:
\begin{itemize}
    \item $u'$ is a working vertex of some structure that is not marked as on-hold or extended.
    \item $\ell(u') = s$.
\end{itemize}
The right part contains all inner vertices $v'$ with $\ell(v') > s + 1$.
The arc set is the set of all arcs in $G'$ from the left part of $H'_s$ to its right part.
Note that the arc set of $H'_s$ is exactly the set of all $s$-feasible arcs in $G'$.
\end{definition}

\begin{algorithm}
\begin{algorithmic}[1]
\medskip 
\Statex \textbf{Input:} a graph $G$, the parameter $\eps$, the current matching $M$, the structure $\cS_\alp$ of each free vertex $\alp$, the set of paths $\calP$
\medskip 
\Statex \hrule 

\For{stages $s = 1, 2, \dots, \lmax$} \label{line:stage-loop}
    construct $H'_s$ as defined in \cref{def:Hps} \label{line:construct-Hps}
    \For {$22 c \ln(1/\eps)$ iterations} \label{line:sim-ext-iteration}
        \State find a matching $M'$ by invoking $\Amat$ on $H'_s$ \label{line:sim-ext-Amat}
        \For {each arc $(u', v')$ in $M'$}
            \State perform $\algOvertake$ on $(u', v')$  \label{line:sim-ext-overtake}
        \EndFor
        \State reconstruct $H'_s$ as defined in \cref{def:Hps} \label{line:reconstruct-Hps}
    \EndFor
    \For {each $(u, v) \in G$ such that $(\Omg(u), \Omg(v))$ is in $H'_s$} \label{line:sim-ext-contaminated}
        \State mark $(u, v)$ as contaminated \label{line:Hps-contaminated} \Comment{This step is only for the analysis}
    \EndFor
\EndFor
\State execute \cref{alg:sim-check-1} to simulate \algCheck \label{line:sim-ext-contract}
\end{algorithmic}
\caption{Simulating $\algExtend$ using $\Amat$.} \label{alg:sim-extend-1}
\end{algorithm}

A stage $s$ consists of $O(c\ln(1/\eps))$ iterations (see \cref*{line:sim-ext-iteration}).
In each iteration, $\Amat$ is invoked on $H'_s$ to find a matching $M'$.
Then, we perform $\algOvertake$ using each matched edge returned by $\Amat$.
These overtaking operations may modify the graph.
We update $H'_s$ to reflect the change and proceed to the next iteration.
The above iterative procedure may not find all $s$-feasible arcs in $H'_s$.
After the last iteration, we mark all remaining $s$-feasible arcs as contaminated.
This completes the description of a stage.

After all stages are executed, \cref*{line:sim-ext-contract} invokes \cref{alg:sim-check-1} to simulate $\algCheck$.
This step is to simulate the $\algAugment$ and $\algContract$ operations performed by $\algExtend$.

\begin{remark}
Recall that the semi-streaming algorithm will call \algCheck immediately after \algExtend (see \cref{alg:phase}).
Therefore, one can also skip the execution of \cref*{line:sim-ext-contract} in \cref{alg:sim-extend-1}, it only causes \algCheck to be executed two times.
We keep \cref*{line:sim-ext-contract} of \cref{alg:sim-extend-1} to make our simulation more similar to the original \algExtend.
\end{remark}

\paragraph{Analysis.}
Since there are $\lmax = O(\eps^{-1})$ stages and each stage calls $\Amat$ $O(c \ln(1/\eps))$ times, the total number of calls is $O(c\eps^{-1} \ln(\eps^{-1}))$.

We proceed to analyze the correctness.
The analysis consists of three parts.
First, we prove that the overtaking operations in \cref{line:sim-ext-overtake} are well-defined; i.e., each operation is performed on a type 3 arc.
More specifically, let $M' = \{e_1, e_2, \dots, e_t\}$ be the matching found in \cref{line:sim-ext-Amat};
we show that for $i = 1, 2, \dots, t$, $e_i$ is still an $s$-feasible arc (and therefore a type 3 arc) after we perform \algOvertake on $e_1, e_2, \dots, e_{i-1}$.
Second, we show the following.

\begin{restatable}{lemma}{extcont} \label{lem:extend-contaminated} After running \cref{alg:sim-extend-1}, there are no arcs of type 1, 2, or 3 in $G$ except the contaminated ones.
\end{restatable}

\noindent This highlights that our simulation is essentially performing the three basic operations in a different order, and it is inexact only because of the contaminated arcs.
Third, we show that the contaminated arcs are negligible, in the sense that they only intersect a small number of augmenting paths in $G$.
This property is used in the full correctness proof in \cref{sec:app-correctness}.

The proof of Parts 1 and 2 is straightforward but technical, and therefore they are deferred to \cref{sec:app-correctness}.

\paragraph{Part 3 of the analysis.}
Consider a fixed iteration.
An important property is that all vertices matched in the iteration will be removed in \cref*{line:reconstruct-Hps} of \cref{alg:sim-extend-1}.
To see this, consider a matched arc $(u', v')$ returned by $\Amat$.
Since the $\cS_u'$ overtakes $\cS_v'$, it is marked as extended;
in addition, the label of $v'$ is updated to be $s + 1$.
By \cref{def:Hps}, $u'$ and $v'$ no longer qualify as vertices of $H'_s$.
That is, when \cref*{line:reconstruct-Hps} reconstructs $H'$, both $u'$ and $v'$ are removed.

An additional observation is that \cref*{line:reconstruct-Hps} never adds new vertices into $H'_s$.
Consequently, throughout all iterations, $H'_s$ is a decremental graph.
We obtain the following using the above observations.

\begin{lemma}
\label{lem:drop2}
After the last iteration of \cref{alg:sim-extend-1}, it holds that $\mu(H') \leq \eps^{22} |M|$.
\end{lemma}
\begin{proof}
In each iteration, \cref*{line:reconstruct-Hps} removes all vertices incident to the $c$-approximate matching $M'$, and no new vertices are added to $H'_s$.
This iterative process is similar to our simulation of \algCheck (\cref{alg:sim-check-1}), where we remove a $c$-approximate matching in each iteration.
By repeating the analysis of \cref{lem:drop} with $H'$ replaced by $H'_s$, we show that $\mu(H'_s)$ is decreased by a factor of $\eps^{22}$ after $22 c \ln n$ iteration.

Recall that the left part $L$ of $H'_s$ consists only of outer vertices.
Since each outer vertex is matched in $M / \Omega$, $L \leq |M / \Omg| \leq |M|$.
Hence, $\mu(H'_s) \leq |M|$ initially, and after all iterations it is reduced to at most $\eps^{22} |M|$.
This completes the proof.
\end{proof}

\begin{lemma} \label{lem:vc2} Let $E_C$ be the set of contaminated arcs marked in \cref{alg:sim-extend-1}. Then, $E_C$ only intersects $2 \eps^{17} |M|$ vertex-disjoint augmenting paths in $G$.
\end{lemma}
\begin{proof}
Let $E_s$ be the set of contaminated arcs marked in \cref{alg:sim-extend-1} in stage $s$.
By using \cref{lem:drop2} and repeating the argument in \cref{lem:vc}, we know that $E_s$ intersects at most $2 \eps^{19} |M|$ vertex-disjoint augmenting paths in $G$.
Since there are $\lmax$ stages, $|E_C| = \sum_s |E_s| \leq \lmax \cdot 2\eps^{19}|M| \leq 2 \eps^{17} |M|$.
This completes the proof.
\end{proof}

\section{Dynamic matching: Boosting framework for induced subgraph oracle} \label{sec:framework-dynamic}

In previous dynamic algorithms, a key subroutine is a sublinear-time algorithm $\Aweak$ for finding a constant-approximate maximum matching in a given vertex-induced subgraph, on the condition that the induced subgraph contains a large matching.

\begin{definition}[Definition of $\Aweak$] \label{def:Aweak}
Given a graph $G = (V, E)$, a subset of nodes $S \subseteq V$, and a parameter $\delta \in (0, 1)$, the algorithm $\Aweak$ returns either $\bot$ or a matching $M$ in $G[S]$ such that $|M| \geq \lambda \cdot \delta n$, where $\lambda \in (0, 1)$ is a constant.
In addition, if $\mu(G[S]) \geq \delta n$, then $\Aweak$ does not return $\bot$.
\end{definition}

In the following, we present an algorithm that finds a $(1 + \eps)$-approximate maximum matching using $\peps \cdot \frac{1}{\lambda}$ invocations of $\Aweak$ and $\tO(n \cdot \peps \cdot \frac{1}{\lambda})$ time, assuming that $\mu(G) = \Omega(\eps n)$.
The algorithm is summarized in the following theorem.

\begin{theorem} \label{thm:framework-2}
Let $\Aweak$ be an algorithm satisfying \cref{def:Aweak}.
There is an algorithm that given a parameter $\eps \in (0, \frac{1}{2})$ and a graph $G$ with $\mu(G) = \Omg(\eps n)$, computes a $(1+\eps)$-approximate maximum matching of $G$ by making $\peps \cdot \frac{1}{\lambda}$ calls to $\Aweak$ on adaptively chosen subsets of vertices with $\delta$ set to be $\poly(\eps)$. The algorithm spends an additional processing time of $O(\peps \cdot n)$ per call.
\end{theorem}

The algorithm is used as a subroutine in our dynamic algorithms.
(See \cref{sec:dynamic}.)

\begin{remark}
In this section, we focus on obtaining an algorithm with $\peps$ dependency in its time complexity.
However, we do not attempt to optimize the degree of the polynomial.
The resulting algorithm has an $O(\eps^{-108})$ dependence on $\eps$ and only works for $n \geq \eps^{-300}$, but the exponents can be greatly reduced by a more careful analysis.
\end{remark}

\subsection{Notations} \label{sec:sim-notation-2}
To prove \cref{thm:framework-2}, we assume that $\mu(G) \geq t \cdot \eps n$ for some constant $t$.
We also assume that $n \geq 1 / \eps^{300}$; otherwise, a $(1+\eps)$-approximate matching can be found by using \cite{DP14}'s algorithm, which takes $O(m \log(1/\eps) / \eps) = O(n^2 \log(1/\eps) / \eps) = O(\peps)$ time.
We fix $\delta = \eps^{107}$ throughout the section.

\paragraph{Representation of the graph.}
We assume Random Access Machine (RAM) as the model of computation.
The algorithm takes the adjacency matrix of $G$ as input.
It stores in memory the tuple $(G_\alp, \Omg_\alp, w'_\alp)$ of each structure $\cS_\alp$.
In addition, it also stores and maintains all in-structure arcs; i.e., the arcs connecting two
$G$-vertices of the same structure.
Since each structure contains at most $\Delta = O(\eps^3)$ vertices and $\Delta^2 = O(\eps^6)$ edges, all structures can be stored using $n \peps$ memory.

\paragraph{An auxiliary graph.} 
Our algorithm simulates the framework in \cref{sec:framework-distr}.
At the beginning of the algorithm, we create a bipartite graph $B$, which is defined as follows:
\begin{definition}[Bipartite graph $B$] \label{def:B}
Let $B$ be a bipartite graph obtained by splitting each $G$-vertex $v$ into two copies $v^+$ and $v^-$, called \emph{outer copy} and \emph{inner copy}, respectively.
The left part is $V^+ = \{v^+ \mid v \in V(G)\}$ and the right part is $V^- = \{v^- \mid v \in V(G)\}$.
The edge set of $B$ is $\{(u^+, v^-) \mid (u, v) \in G\} \cup \{(v^+, u^-) \mid (u, v) \in G\}$.
Intuitively, the vertex $v^+$ represents that $v$ is an outer vertex, and $v^-$ represents that $v$ is an inner or unvisited vertex.    
\end{definition}

In our algorithm, each invocation of $\Aweak$ will be on either $B$ or $G$.
The graph $B$ is not constructed explicitly.
(Explicit construction of $B$ requires $\Omega(m)$ time, whereas our goal is a $O(\peps \cdot n)$-time algorithm.)
Instead, we only store the vertex set of $B$.
Still, we can assume access to the adjacency matrix of $B$, because it can be supported by making $O(1)$ queries to $G$.

% The following lemma shows an enhanced version of $\Aweak$ obtained by invoking $\Aweak$ $\peps$ times.

% \begin{lemma}
% Given a graph $G = (V, E)$, a subset of nodes $S \subseteq V$, and a parameter $\delta' \in (0, 1)$, we can find a matching $M$ in $G[S]$ such that $|M| \geq \frac{\delta'}{4} \cdot n$ assuming that $\mu(G[S]) \geq \delta' n$.
% \end{lemma}
% \begin{proof}
% We initialize $M$ as an empty matching and repeat the following: Invoke $\Aweak$ on $S - V(M)$ with the parameter $\delta = \frac{\delta'}{4}$. If $\Aweak$ returns a matching, we add it to $M$; otherwise, we terminate the procedure.

% Assume that $G[S]$ contains a matching $M'$ of size $\delta' n$.
% We claim that $|M| \geq \frac{\delta'}{4} n$ after the termination of the procedure.
% Suppose, by contradiction, that $|M| < \frac{\delta'}{4} n$.
% Then the matched vertices in $M$ can intersect at most $\frac{\delta'}{2} n$ edges in $M'$.
% This implies that at least $\frac{\delta'}{2} n$ edges of $M'$ have both endpoints in $S - V(M)$, and thus $\Aweak$ could not return $\bot$ in the last invocation, which is a contradiction.
% Consequently, $|M| \geq \frac{\delta'}{4} n$.

% !!!Parameter, Complexity....
% \end{proof}

\subsection{Concentration bounds}
We use the following concentration bounds.
Let $\mathcal{X} = \{X_1, X_2, \dots, X_k\}$ be $\{0, 1\}$-random variables.
Let $X = \sum_{i=1}^{k} X_i$ and $\mu = \bbE[X]$.

\begin{lemma}[Chernoff bound] \label{lem:Chernoff}
If $X_1, X_2, \dots, X_k$ are mutually independent, then
$\Pr[X \leq (1 - \gamma)\mu] \leq e^{-\mu \gamma^2 / 2}$.
\end{lemma}

\begin{definition}[Dependency graph] \label{def:depencency-graph}
A \emph{dependency graph} for $\mathcal{X}$ has vertex set $[k]$ and an edge set such that for each $i \in [k]$, $X_i$ is mutually independent of all other $X_j$ such that $\{i, j\}$ is not an edge.
\end{definition}

For any non-negative integer $d$, we say that the $X_i$’s exhibit \emph{$d$-bounded dependence}, if the $X_i$’s have a dependency graph with maximum vertex degree $d$.

\begin{lemma}[\cite{Pemmaraju2001}] \label{lem:limited-dependence}
If $X_1, X_2, \dots, X_k$ exhbit $d$-bounded dependence, then
$\Pr[X \leq (1 - \gamma)\mu] \leq \frac{4(d+1)}{e} \cdot e^{\frac{-\mu \gamma^2}{2(d+1)}}$.
\end{lemma}

\subsection{Overview of the simulation}

In the following, we present an overview of our simulation except for $\algExtend$ and $\algCheck$, which are the only two procedures we do not simulate exactly.

\paragraph{Simulation of the outer loop (\cref{alg:outline}).}
\cref*{line:2-approx-matching} of \cref{alg:outline} computes an initial matching $M$.
We simulate this step using the approach in \cref{lem:init}, that is, to repeatedly find a matching on the set of unmatched vertices.
A simple analysis shows that $O(1/(\delta\cdot\lambda))$ calls to $\Aweak$ suffices (see \cref{sec:sim-init-2} for details).

Consider the for-loop in \cref*{line:scale-h}.
We will describe shortly the simulation for \cref{alg:phase}.
Recall that \cref{alg:phase} removes some vertices during the execution.
For each vertex of $G$, we store a label indicating whether it is removed in \cref{alg:phase}.
Hence, \cref*{line:restore} can be done in $O(n)$ time by setting all vertices as not removed.
\cref*{line:augment} can be done in $O(n)$ time because each of $M$ and $\calP$ contains at most $n$ edges.
Since the loop has $\peps$ iterations, the total time spent excluding the execution of \cref{alg:phase} is $O(n \peps)$.

\paragraph{Simulation of a phase (\cref{alg:phase}).}
Throughout the algorithm, we maintain the structure of each free vertex and the set $\Omega$ of blossoms, which takes $O(n \peps)$ space.
Lines 1-4 of \cref{alg:phase} can be easily simulated in $O(n)$ time.
(Line 2 takes $O(n)$ time because $M$ contains $O(n)$ arcs.)
Consider the for-loop in Lines 5-12.
Lines 5-9 and Line 12 can be simulated in $O(n)$ time by examining the structure of each free vertex.
Lines 10 and 11 are the only two procedures that require a pass over the stream, which we cannot simulate exactly because it takes $\Omg(m)$ time.
The simulation of Lines 10 and 11 is described in, respectively, \cref{sec:sim-extend-1,sec:sim-check-1}.

\subsection{Computation of the initial matching} \label{sec:sim-init-2}
The first step of the algorithm is to compute a constant approximate matching.

\begin{lemma} \label{lem:init-2}
Assume that $\mu(G) \geq d\eps n$ for some constant $d$.
A $3$-approximate matching can be computed in $O(n / \eps)$ time plus $O(1/(\delta \lambda))$ calls to $\Aweak$ with $\delta \leq \frac{d\eps}{3}$.
\end{lemma}
\begin{proof}
The computation is done by iteratively finding matchings in $G$.
We initialize $M$ as an empty matching.
In each iteration, we invoke $\Aweak$ on the set of unmatched vertices with  $\delta = \frac{d\eps}{3}$.
All matched edges returned by $\Aweak$ are added to $M$.
The procedure is repeated until $\Aweak$ returns $\bot$.

Since $\Aweak$ must find a matching of size $\lambda\delta n$ in each iteration except the last, there are at most $1 / (\lambda \delta)$ iterations.
Each iteration spends $O(n)$ time to update $M$ and uses 1 invocation of $\Aweak$.
The claimed running time follows.

We complete the proof by showing that $M$ is a 3-approximation.
Consider the matching $M$ after the last iteration.
Let $M^*$ be the maximum matching in $G$.
Since $\Aweak$ returns $\bot$ in the last iteration, there are at most $\delta n$ edges in $M^*$ whose endpoints do not intersect $M$.
Since $M$ only contains $2|M|$ endpoints, we have
\[
|M^*| \leq 2 |M| + \delta n = 2 |M| + \frac{d\eps}{3} n \leq 2|M| + \frac{1}{3}|M^*|,
\]
where the last inequality comes from $|M^*| \geq d \eps n$.
It follows that $|M| \geq \frac{1}{3}|M^*|$, concluding the proof.
\end{proof}

\subsection{Simulation of $\algCheck$} \label{sec:sim-check-2}
\paragraph{Overview.}
We implement the simulation in \cref{sec:sim-check-1} differently.
Recall that \algCheck has two steps (see \cref{sec:check}), which find arcs for $\algContract$ and $\algAugment$, respectively.
The first step can be done by scanning all in-structure edges and performing $\algContract$ whenever possible.
Since each structure and blossom contains at most $\Delta^2$ edges, this step can be done in $O(n\Delta^2)$ = $O(n \eps^{-6})$ time.

Recall that in \cref{sec:sim-check-1}, Step 2 is done by finding matching on a new graph, in which each vertex is a structure and each edge represents a possible augmentation between two structures.
Implementing this approach with $\Aweak$ is challenging because this new graph is not a vertex-induced subgraph of $G$ or $B$.
To address this, we use the idea of random sampling:
Instead of contracting each structure into a vertex, we sample a vertex from each structure and invoke $\Aweak$ on the sampled vertices.
Let $S$ denote the set of sampled vertices.
If an edge $e = (u, v)$ connects outer vertices of two structures in $G$, there is a probability of at least $1 / \Delta^2$ that $e$ appears in $G[S]$.
Therefore, $\bbE[\mu(G[S])] \geq \mu(G) / \Delta^2$.
By adapting the analysis in \cref{lem:drop}, we show that $\mu(G)$ drops exponentially in expectation.
\cref{fig:sampling} gives an example.

\begin{figure}[h] 
    \centering
    \includegraphics[width=0.5\textwidth]{Sketches/sampling.png}
        \caption{An illustrative example for the sampling procedure, in which the two sampled vertices $u$ and $v$ are marked in gray. The structures $\cS_\alp$ and $\cS_\beta$ are connected by two type 2 arcs $(u, v)$ and $(x, y)$, and the goal of the procedure is to preserve one of these arcs in $G[S]$.
        In this example, both $u$ and $v$ are sampled, and thus $(u, v)$ is contained in $G[S]$ as desired.
        In general, each of $(u, v)$ and $(x, y)$ is contained in $G[S]$ with probability at least $1/\Delta^2$.
        }
        \label{fig:sampling}
\end{figure}

\paragraph{Simulation procedure for Step 2.}
The simulation consists of $I = 1 / (2\lambda\eps^{107}) + 1$ iterations.
In an iteration $i$, we sample one outer $G$-vertex $v$ from each structure.
Let $S$ be the set of vertices that are sampled.
We invoke the $\Aweak$ on $S$ with $\delta = \epsilon^{107}$ to obtain a matching $N$ on $S$.
Note that each edge in $N$ connects outer vertices of two different structures.
In addition, each structure contains at most one vertex matched in $N$.
We perform $\algAugment$ on each edge in $N$, which removes all involved structures.
This completes the description of an iteration.

\paragraph{Time Complexity.}
For each iteration, the algorithm spends $O(n)$ time for sampling vertices.
Then, $\Aweak$ is invoked once to find the matching $N$.
Let $(u, v)$ be an arc in $N$.
Let $\alp$ (resp. $\beta$) be the free vertex such that $\cS_\alp$ contains $u$ (resp. $\cS_\beta$ contains $v$).
In the following, we show that performing $\algAugment$ on $(u, v)$ can be done in $O(|\cS_{\alp}| + |\cS_{\beta}|)$ time.
First, we find the augmenting path $P'$ in $G'$ between $\Omg(\alp)$ and $\Omg(\beta)$ by concatenating the tree paths from $\Omg(\alp)$ to $\Omg(u)$ and from $\Omg(\beta)$ to $\Omg(v)$.
Then, we use \cref{lem:even-path} to transform this $P'$ to $P$.
Since all vertices in $P$ are in $\cS_\alp \cup \cS_\beta$, this step, this computation takes $O(|\cS_\alp| + |\cS_\beta|)$ time.
Then, we mark all vertices in $\cS_\alp$ and $\cS_\beta$ as
removed, which can also be done in $O(|\cS_\alp| + |\cS_\beta|)$ time.
Hence, the $\algAugment$ operation requires $O(|\cS_\alp| + |\cS_\beta|)$ time.
Since the structures are vertex-disjoint, performing $\algAugment$ on all structures matched in $N(i)$ takes $O(n)$ time.

Recall that Step 1 of $\algCheck$ requires $O(n / \eps^6)$ time.
We conclude that the simulation procedure takes $O(n / \eps^6)$ time in total.

\paragraph{Correctness.}
Consider a fixed iteration $i$.
The matching $N$ returned by $\Aweak$ is a set of edges connecting outer vertices of different structure.
Thus, $\{(\cS_u, \cS_v) \mid (u, v) \in N\}$ is a matching in $H'$ (see \cref{def:Hp}).
Consequently, the simulation can be seen as a different implementation of \cref{alg:sim-check-1}, where in each iteration we find a matching in $H'$ using $\Aweak$ instead of $\Amat$.
We now show that $I$ iterations suffice.

Let $e'$ be an edge in $H'$.
By definition, $e'$ corresponds to at least an arc $(u, v) \in G$, which is an arc satisfying $(\cS_u, \cS_v) = e'$.
We say $e'$ is \emph{preserved} if $(u, v)$ is sampled.
Since each vertex is sampled with probability at least $1 / \Delta$, $e'$ is preserved with probability at least $1/\Delta^2 = \eps^6$.

Let $N'$ be a maximum matching in $H'$ at the beginning of iteration $i$.
We first show that if $N'$ is large, then our sampling preserves a large matching.

\begin{lemma} \label{lem:check-prob}
Consider a fixed iteration $i$. If $|N'| \geq \eps^{100} n$ at the beginning of $i$, then $G[S]$ contains a matching of size $\eps^{107} n$ with probability at least $1 - n^{-10}$.
\end{lemma}
\begin{proof}
For each edge $e' \in N'$, define $X_{e'}$ as the indicator random variable of whether $e'$ is preserved.
Let $X = \sum_{e' \in N'} X_{e'}$ denoting the number of edges in $N'$ that are preserved.
By our discussion above, $\bbE[X_{e'}] \geq \eps^6$.
Thus, we have
\[
    \bbE[X] = \sum_{e' \in N^*(i)} \bbE[X_{e'}] \geq |N'| \cdot \eps^6 \geq \eps^{106} n.
\]
In addition, since $N'$ is a matching, the random variables $\{X_{e'} \mid e' \in N^*(i)\}$ are mutually independent.
Applying a Chernoff bound (\cref{lem:Chernoff}) with $\gamma = 0.5$, we obtain $\Pr[X \leq 0.5 \bbE[X]] \leq e^{-\bbE[X] \gamma^2 / 2} \leq e^{-\eps^{106} n / 8} \leq e^{-10\ln{n}} = n^{-10}$ for large enough $n$.
(Recall that we have assumed $n \geq 1 / \eps^{300}$.)
Therefore, with probability at least $1 - n^{-10}$, at least $0.5 \bbE[X] \geq \eps^{107} n$ edges in $N'$ are preserved.
This completes the proof.
\end{proof}

% Let $G(i)$ be the graph $G$ at the beginning of $i$-th iterations (which consists of the vertices that have not been removed).
% Let $N^*(i)$ be the set of edges obtained as follows:
% Make a pass over the edges of $G(i)$; if an edge $e$ connects outer $G$-vertices of two different structures, add $e$ to $N^*(i)$ and perform $\algAugment$ on it.
% That is, $N^*(i)$ is the set of edges on which the semi-streaming algorithm (\cref{alg:outline}) would perform $\algAugment$ when executing $\algCheck$.
% Note that $N^*(i)$ forms a matching of $G(i)$, and each structure contains only one vertex in $N^*(i)$.
% For ease of notation, we denote by $G(I+1)$ the graph $G$ after the last iteration.
% (Recall that $I$ is the total number of iterations.)
% Define $N^*(I+1)$ similarly.
% Our goal is to show that $|N^*(I+1)| \leq \eps^{100} n$.

% (That is, the induced subgraph $G(i)[S(i)]$ contains $e$.)
% We first show that if $N^*(i)$ is large for an iteration $i$, then our sampling preserves a large matching.

\begin{claim} \label{clm:missed-aug-2}
At the end of all iterations, $|N'| \leq \eps^{100} n$ holds with high probability.
\end{claim}
\begin{proof}
We say an iteration $i$ ($1 \leq i \leq I$) is \textit{successful} if at least one of the following holds:
\begin{enumerate}
    \item[] (E1) $\Aweak$ finds a matching with at least $\lambda \cdot \eps^{107} n$ edges in iteration $i$;
    \item[] (E2) $|N^*(i)| \leq \eps^{100} n$.
\end{enumerate}
By \cref{lem:check-prob}, if (E2) does not hold, then (E1) holds with probability $1 - n^{-10}$.
Hence, each iteration is successful with probability at least $n^{-10}$, regardless of the outcome of previous iterations.
By a union bound, with probability at least $1 - I \cdot n^{-10} \geq 1 - n^{-9}$, all iterations are successful.

Each time $\Aweak$ finds a matching $N$, we remove $2|N|$ structures from $G$.
Therefore, throughout all iterations at most $n$ edges are found.
Thus, (E1) can happen at most $1 / (2 \lambda \cdot \eps^{107}) < I$ times.
Hence, (E2) holds at the last iteration with probability at least $1 - n^{-9}$.
This completes the proof of the claim.
\end{proof}

As in \cref{sec:sim-check-1}, all type 1 arcs in $G$ are said to be contaminated after the last iteration.
By \cref{clm:missed-aug-2} and since $|M| = \Omg(\eps n)$, $\mu(H')$ is at most $\eps^{100} n \leq \eps^{20}|M|$.
By an argument similar to the proof of \cref{lem:vc}, the contaminated arcs intersect at most $\eps^{17} |M|$ vertex-disjoint augmenting paths.
This shows that our simulation $\Aweak$ can be seen as a different implementation of \cref{alg:sim-check-1}.
Hence, the correctness proof in \cref{sec:app-correctness} applies to the simulation as well.

\subsection{Simulation of $\algExtend$} \label{sec:sim-extend-2}
Recall that $\algExtend$ scans through all arcs in the stream, and performs \algContract, \algAugment, and \algOvertake whenever possible.
To simulate $\algExtend$, we follow the approach in \cref{sec:sim-extend-1}.
First, we repeatedly find type 3 arcs to perform $\algOvertake$.
Then, we invoke our simulation for $\algCheck$ in \cref{sec:sim-check-2}.

Recall that the finding of type 3 arcs is done by executing $\lmax$ stages, where in stage $s$, we find a matching consisting of $s$-feasible arcs.
In \cref{sec:sim-extend-1}, each stage consists of finding matchings in a derived bipartite graph $H'_s$.
(See \cref{def:Hps} for the definition of $H'_s$.)
We simulate this approach by finding matchings in $B$.
(See \cref{def:B} for the definition of $B$.)

\paragraph{Simulation procedure.}
Consider a fixed stage $s$.
The simulation of stage $s$ consists of $I = 1 / (2\lambda \eps^{100}) + 1$ iterations.
The overall idea is to repeat the procedure and analysis in \cref{alg:sim-extend-1}.
However, for technical reasons, we need to handle ``in-structure overtake'' separately.
More precisely, we maintain the following invariant throughout all iterations.

\begin{invariant} \label{inv:no-in-struct}
At the beginning of each iteration, there are no $s$-feasible arcs that connect two vertices of the same structure.
\end{invariant}

\noindent To maintain \cref{inv:no-in-struct}, the following is performed before the first iteration and after each iteration:
We scan through each arc connecting two vertices of the same structure;
if the arc is $s$-feasible, we perform \algOvertake on it and mark the structure containing the arc as extended.
After this step, \cref{inv:no-in-struct} holds.

Each iteration finds a set of type 3 arcs for cross-structure overtake.
First, we construct a vertex subset $S \subseteq V(B)$ as follows.
We sample one $G$-vertex from each structure uniformly at random.
Let $T$ denote the sampled vertices.
For each outer vertex $u$ in $T$, we add $u^+$ (the outer copy of $u$ in the graph $B$) in $S$ if and only if the following holds:
\begin{itemize}
    \item $\cS_u$ is not on-hold or extended,
    \item $\Omg(u)$ is the working vertex of $\cS_u$, and
    \item $\ell(\Omg(u)) = s$.
\end{itemize} 
For each inner vertex $v$ in $T$, we add $v^-$ in $S$ if and only if $\ell(\Omg(v)) > i + 1$.
It is not hard to see that for each arc $(u^+, v^-) \in B$, its corresponding arc $(u, v) \in G$ is of type 3.
We invoke $\Aweak$ on $S$ with $\delta = \eps^{107}$ to find a matching $N_B \subseteq E(B)$.
Let $N' = \{(\Omg(u), \Omg(v)) \mid (u^+, v^-) \in N_B \}$.
Since each structure is only adjacent to at most one arc in $N_B$, no two vertices in $N'$ may share an endpoint; that is, $N'$ is a matching of $H'_s$.
We perform $\algOvertake$ on each arc in $N'$.
All type 3 arcs remaining in $G$ after $I$ iterations are considered contaminated.

\paragraph{Analysis.}
Consider a fixed iteration $i$ in some stage $s$.
Recall that our simulation finds a matching $N_B$ by invoking $\Aweak$.
The first step is to show that if $H'_s$ contains a large matching, then $\Aweak$ succeeds with high probability.
Let $N^*$ be a maximum matching in $H'_s$ at the beginning of iteration $i$.

\begin{lemma} \label{lem:extend-prob}
Consider a fixed iteration $i$ in a stage $s$. At the beginning of iteration $i$, if $|N^*| \geq \eps^{100} n$, then $B[S]$ contains a matching of size $\eps^{107} n$ with probability at least $1 - n^{-10}$.
\end{lemma}
\begin{proof}
Let $N^*_B$ be a matching in $B$ constructed as follows:
For each arc $(u', v') \in N^*$, find an arc $(u^+, v^-) \in B$ such that $\Omg(u) = u'$ and $\Omg(v) = v'$ (if there are multiple such $(u^+, v^-)$, pick any of them).
Note that $|N^*_B| = |N^*| \geq \eps^{100} n$.

We say an arc $(u^+, v^-) \in N^*_B$ is \emph{preserved} if both of its endpoints are in $S$.
This happens exactly when we sample $u$ from $\cS_u$ and $v$ from $\cS_v$.
By \cref{inv:no-in-struct}, $\cS_u \neq \cS_v$, and hence 

\begin{equation} \label{eqn:pr-preserved}
\Pr[\text{$(u^+, v^-)$ is preserved}] = \Pr[\text{$u$ and $v$ are both sampled}] = \frac{1}{|S_u|} \cdot \frac{1}{|S_v|} \geq 1 / \Delta^2 = \eps^6,
\end{equation}
where the last inequality is because each structure contains at most $\Delta$ vertices.

For each arc $e \in N_B^*$, define $X_e$ as the indicator random variable of whether $e$ is preserved.
Let $X = \sum_{e \in N_B^*(i)} X_e$ denoting the number of edges in $N_B^*$ that are preserved.
Consider an edge $(u, v) \in N_B^*(i)$.
By \cref{eqn:pr-preserved}, $\bbE[X] = \sum_{e \in N_B^*(i)} \bbE[X_e] \geq |N_B^*(i)| \cdot \eps^6 \geq \eps^{106} n$.

Let $\cX = \{X_e \mid e \in N^*_B\}$.
The random variables $\cX$ may not be mutually independent.
In particular, two arcs $(a^+, b^-)$ and $(u^+, v^-)$ are dependent if and only if $(a, b)$ and $(u, v)$ are adjacent to a common structure.
Since $N^*$ is a matching, each structure is adjacent to at most $\Delta$ arcs in $N^*$.
Therefore, each $X_e$ depends on at most $2\Delta$ other random variables in $\cX$.
(More formally, $X_e$ is mutually independent of the set of all $X_{e'}$ such that $e'$ does not share a structure with $e$.)
As a result, $\cX$ admits a dependency graph (\cref{def:depencency-graph}) with maximum degree $2\Delta$.

Applying the concentration bound for limited dependence (\cref{lem:limited-dependence}) with $\gamma = 0.5$ and $d = 2\Delta$, we obtain
\[
\Pr[X \leq 0.5 \bbE[X]] \leq \frac{4(d+1)}{e} \cdot e^{-\bbE[X] \gamma^2 / (2d+2)} \leq e^{-\eps^{106} n / (16\Delta + 8)} = e^{\Theta(\eps^{112} n)} \leq e^{-10\ln{n}} = n^{-10},
\]
where the last inequality holds for a large enough $n$.
(Recall that we have assumed $n \geq 1 / \eps^{300}$.)
Therefore, with probability at least $1 - n^{-10}$, at least $0.5 \bbE[X] \geq \eps^{107} n$ edges in $N^*(i)$ are preserved.
This completes the proof.
\end{proof}

\begin{lemma}
Let $E_C$ be the set of contaminated edges created by simulating $\algExtend$ using $\Aweak$.
With high probability, $E_C$ admits a vertex cover of size $O(\eps^{92} n)$.
\end{lemma}
\begin{proof}
Consider a fixed stage $s$.
As argued in \cref{sec:sim-extend-1}, after we perform \algOvertake on an arc $(u', v')$ of $H'_s$, its two endpoints are no longer part of $H'_s$ (because $S_{u'}$ is marked as extended, and $\ell(v')$ is updated to $s+1$).
Hence, whenever $\Aweak$ finds a matching, the number of vertices in $H'_s$ decreases by $2\lambda\delta n$.
Since $H'_s$ has at most $n$ vertices at the beginning of $s$, there can only be $1 / (2\lambda\delta) < I$ iterations where $\Aweak$ finds a matching.
By an argument similar to \cref{clm:missed-aug-2}, we can show that $|N^*| \leq \eps^{100} n$ holds with high probability after $I$ iterations.
This implies that the contaminated edges created in stage $s$ admit a vertex cover of size $\eps^{100} n \times 2\Delta \leq \eps^{93} n$.
Since there are $\lmax = O(1 / \eps)$ stages, the set of all contaminated edges can be covered with $O(\eps^{92} n)$ vertices.
\end{proof}

\section{Applications in the dynamic setting} \label{sec:dynamic}
In this section, we present new algorithms for maintaining a $(1+\eps)$-approximate matching in the \emph{fully dynamic setting}, using the framework in \cref{thm:framework-2}.

\subsection{Dynamic settings}

In the \emph{dynamic $(1+\eps)$-approximate matching problem}, the task is to maintain a $(1+\eps)$-approximate matching while the graph undergoes edge updates.
We will focus on the \emph{fully dynamic setting} where the graph undergoes both edge insertions and deletions over time.
The edge updates are given online: each update must be processed before the next update is given.
We denote by $G = (V, E)$ the input graph, $n = |V|$, and $m$ the maximum number of edges in $G$ as it undergoes edge updates.
%We say a dynamic matching algorithm has amortized update time $T$ if there is a fixed polynomial $p(n)$ such that the algorithm's running time on the first $i$ edge updates is bounded by $i \cdot T$ for all $i \geq p(n)$. 
We say a dynamic matching algorithm has amortized update time $T$ if the algorithm's running time on the first $i$ edge updates is at most $i \cdot T$.
Our approach uses the pre-processing time of $n \poly \log n$.
A dynamic $(1+\eps)$-approximate matching algorithm is said to succeed with probability $p$ if it maintains a matching $M$ satisfying the following:
After any fixed edge update in the update sequence, $M$ is a $(1+\eps)$-approximate maximum matching with probability at least $p$.

\subsection{A framework for dynamic $(1+\eps)$-approximate matching}
Recent works \cite{BhattacharyaKS23,BG24,assadi2024improved} have shown that the fully dynamic $(1+\eps)$-matching problem reduces to the following problem with a specific set of parameters.

\begin{minipage}{0.95\linewidth}
    \begin{mdframed}[backgroundcolor=gray!15, linecolor=red!40!black]
    \begin{problem} \label{prob:framework}
    The problem is parameterized by a positive integer $q \geq 1$ and reals $\lambda,\delta,\alpha \in (0,1)$.
    
    \medskip
    
    \emph{\textbf{Input:}} a fully dynamic $n$-vertex graph $G=(V,E)$ that starts empty, i.e., has $E = \emptyset$, and throughout, never has more than $m$ edges, nor receives more than $\poly(n)$ updates in total. 
	
	\medskip
	
	\emph{\textbf{Updates:}}  The updates to $G$ happen in \emph{\textbf{chunks}} $C_1,C_2,\ldots$, each consisting of exactly $\alpha \cdot n$ edge insertions or deletions in $G$. 

    \medskip
    
	\emph{\textbf{Queries:}} After each chunk, there will be at most $q$ queries, coming one at a time and in an adaptive manner (based on the answer to all prior 
	queries including the ones in this chunk).  Each query is a set $S \subseteq V$ of vertices;
    the algorithm should respond to the query with the guarantee specified in \cref{def:Aweak};
    that is, it returns either a matching in $G[S]$ of size at least $\lambda \cdot \delta n$ or $\bot$;
    furthermore, if $\mu(G[S]) \geq \delta n$, the algorithm does not return $\bot$.
	
	For ease of reference, we list the parameters of this problem and their definitions: 
	\begin{align*}
		n &: \text{number of vertices in the graph}; \\
		m &: \text{maximum number of edges at any point present in the graph}; \\
		q &: \text{number of adaptive queries made after each chunk}; \\
		\lambda &: \text{approximation ratio of the returned matching for each query}; \\
		\delta &: \text{a lower bound on the fraction of vertices matched in the subgraph of $G$ for the query}; \\
		\alpha &: \text{a parameter for determining the size of each chunk as a function of $n$}.
	%	&\textbf{assumption:} we further assume both $\lambda,\delta \leq \alpha$.
	\end{align*}
	
	\noindent
	For technical reasons, we allow additional updates, called \emph{\textbf{empty updates}} to also appear in the chunks but these ``updates'' do not change any edge of the graph, although they will be 
	counted toward the number of updates in their chunks\footnote{This is used for simplifying the exposition when solving this problem recursively; these empty updates will still be counted when computing the amortized runtime of these recursive algorithms.}. 
    \end{problem}
    \end{mdframed}
\end{minipage}

\noindent
It is known that dynamic $(1+\eps)$-approximate matching reduces to $\poly(\log n / \eps)$ instances of \cref{prob:framework} parameterized by any $\lambda \in (0, 1]$, $q = (1/(\lambda\cdot\eps))^{O(1/(\lambda\cdot\eps))}$, $\delta = (\lambda\cdot\eps)^{O(1/(\lambda\cdot\eps))}$, and $\alp = \eps^2$ \cite{BhattacharyaKS23,BG24,assadi2024improved}.
This result provides a framework for solving dynamic $(1+\eps)$-approximate matching, showing that it reduces to implementing $\Aweak$ for a fully dynamic graph.
In the following, we show a reduction with improved parameters $q = 1/\lambda \cdot \peps$, $\delta = \poly(\eps)$.
All other parameters and the number of instances of \cref{prob:framework} remain the same.

We say an algorithm for \cref{prob:framework} has amortized update time $T$ if the algorithm's running time is at most $i \cdot T$ for answering all queries associated with the first $i$ edge updates.
We allow this algorithm to have $n \poly \log n$ preprocessing time.

\begin{theorem} \label{thm:framework}
Let $\eps \in (0, \frac{1}{4}]$ be a parameter.
There exist polynomials $f(\eps)$ and $g(\eps)$ of $\eps$ such that the following holds for all $\lambda \in (0, 1]$.
Suppose that there is an algorithm $\bbA$ for \cref{prob:framework} parameterized by $\lambda$ and $q = \frac{1}{\lambda \cdot f(\eps)}, \delta = g(\eps), \alp = \eps^2$, and with high probability, $\bbA$ takes $\cT(n,m,q,\lambda,\delta,\alpha)$ amortized time to answer all queries.
% (Without loss of generality, we assume that $\cT$ is non-decreasing with $n, m$ and non-increasing with $\eps, \lambda$.)
Then, there is an algorithm that, with high probability, maintains a $(1+\eps)$-approximate matching in a $n$-vertex fully dynamic graph with $\cT(n,m,q,\lambda,\delta,\alpha) \cdot \poly(\log(n) / \eps)$ amortized update time.
\end{theorem}
\begin{proof}
The proof of this theorem is a simple adaptation of the proof of \cite[Theorem 1]{assadi2024improved}.
It directly follows by using \cref{thm:framework-2} to replace \cite[Proposition 2.2]{assadi2024improved} in their proof, so that the resulting time complexity has a polynomial dependence on $1/\eps$.
%We defer the proof to the full version.

% The algorithm for dynamic $(1+\eps)$-approximate matching, whch we denote by algorithm $\bbB$, is constructed as follows.
% First, we use the algorithm in \cref{lem:vertex-reduction} compute $K = \tO(\log n)$ contractions $\phi_1, \phi_2, \dots, \phi_K$.

% Let $e_1, e_2, \dots$ be the input sequence of edge updates on $G$.
% Consider a fixed contraction $\phi_i$.
% The algorithm maintains a matching $M_i$, initially empty.
% Let $G_i$ be the corresponding graph of the contraction and let $s_i$ be the number of vertices in $G_i$.
% The algorithm $\bbB$ reads the sequence of edge updates and transforms each edge update $e_j = (u_j, v_j)$ into an edge update $(\phi_i(u_j), \phi_i(v_j))$ for $G_i$.
% The adjacency matrix, adjacency lists, and vertex degrees of $G_i$ is maintained as described in \cref{sec:dyn-representation}.
% $\bbB$ also runs a copy of algorithm $\bbA$, denoted by $\bbA_i$, with input graph $G_i$.
% The update chunks and queries for $\bbA_i$ are specified as follows.
% Every $\eps^2 n / 2$ edge is grouped as an update chunk and sent to $\bbA$.
% After sending a chunk

% {\color{blue}More coming soon. 50\$ for 7-day early access.}
\end{proof}

\subsection{Dynamic matching via ordered \rs graphs}
\cite{assadi2024improved} had a dynamic $(1+\eps)$-approximate matching algorithm based on its connection with \emph{ordered \rs (ORS) graphs}.
Their result is as follows.

\begin{definition}[Ordered \rs Graphs \cite{BG24}]\label{def:ORS}
	A graph $G=(V,E)$ is called an \emph{$(r,t)$-ORS graph} if its edges can be partitioned into an \emph{ordered set} of $t$ matchings $M_1, M_2, \dots, M_t$ each of size $r$, such that for every $i \in \{1, 2, \dots, t\}$, 
	the matching $M_i$ is an induced matching in the subgraph of $G$ on $M_{i} \cup M_{i+1} \cup \dots \cup M_t$. 

	We define $\ORS{n}{r}$ as the largest choice of $t$ such that an $n$-vertex $(r,t)$-ORS graph exists. 
\end{definition}

\begin{lemma}[\cite{assadi2024improved}]\label{lem:AKS}
There exists an absolute constant $c \geq 1$ such that the following holds.
For any $k \geq 1$, there is an algorithm $\bbA_k(n,m,q,\lambda,\delta,\alpha)$ for \cref{prob:framework} that with high probability takes	
	\[
		O \left( (2q)^{k-1} \cdot \left(\frac{m}{n} \right)^{1/(k+1)} \cdot  \ORS{n}{\lambda \cdot \delta n/2}^{1-1/(k+1)} \cdot n^{6\lambda} \cdot (\log{(n)}/\delta)^{c} \right),
	\]
amortized time over the updates to answer all given queries.  The algorithm works as long as $\lambda < (1/12)^{k}$ and $\alpha \geq \lambda \cdot \delta$. 
\end{lemma}

We obtain the following by combining \cref{thm:framework,lem:AKS}.

\begin{theorem}\label{thm:AKS-dynamic}
	Let $\eps \in (0,1/4)$ be a given parameter, $k \geq 1$ be any integer.
    There exists an algorithm for maintaining a $(1+\eps)$-approximate maximum matching
	in a fully dynamic $n$-vertex graph that starts empty with amortized update time of
	\[
         O \left( n^{1/(k+1)} \cdot  \ORS{n}{\frac{\poly \eps}{15^k} \cdot n}^{1-1/(k+1)} \cdot n^{10/15^k}\right) \cdot \eps^{-O(k)}
	\]	
	The guarantees of this algorithm hold with high probability even against an adaptive adversary. 
\end{theorem}
\begin{proof}
Let $f(\eps), g(\eps)$ be polynomials defined in \cref{thm:framework}.
Let $a, b > 0$ be constants such that $f(\eps) \geq \eps^a$ and $g(\eps) \geq \eps^b$ for $\eps \leq \frac{1}{4}$.
We choose the parameters
\begin{align*}
    q = \frac{1}{\lambda f(\eps)} \leq \frac{1}{\lambda \eps^a}, && \lambda = \frac{1}{15^k}, && \delta = g(\eps) \geq \eps^b, && \alp = \eps^2.
\end{align*}
(As required in \cref{lem:AKS}, $\lambda \leq \frac{1}{12^k}$ and $\alp \geq \lambda \delta$ because $\delta = \eps^{107}$ in our proof of \cref{thm:framework-2}.)
By \cref{lem:AKS}, there is an algorithm for \cref{prob:framework} with the set of parameters running in an amortized update time of
\begin{align*}
    &O \left( (2q)^{k-1} \cdot \left(\frac{m}{n} \right)^{1/(k+1)} \cdot  \ORS{n}{\lambda \cdot \delta n/2}^{1-1/(k+1)} \cdot n^{6\lambda} \cdot (\log{(n)}/\delta)^{c} \right) \\ 
    =~&  O \left( \paren{\frac{2 \cdot 15^k}{\eps^a}}^{k-1} \cdot n^{1/(k+1)} \cdot  \ORS{n}{\frac{2\eps^b}{15^k}}^{1-1/(k+1)} \cdot n^{6/15^k} \cdot (\log{(n)}/\eps^b)^{c} \right), \\
    =~&  O \left( n^{1/(k+1)} \cdot  \ORS{n}{\frac{\poly \eps}{15^k} \cdot n}^{1-1/(k+1)} \cdot n^{10/15^k}\right) \cdot \eps^{-O(k)},
\end{align*}
where we use the fact that $m / n \leq n$ and $\log^c n \leq n^{4/15^k}$.
This completes the proof.
\end{proof}

\cref{thm:AKS-dynamic} improves \cite{assadi2024improved}'s result by reducing the dependence in time complexity on $\eps$, from $\eps^{-O(k/\eps)}$ to $\eps^{-O(k)}$.
The dependence is polynomial for any fixed $k$.
Also, the dependence in the $\ORS{\cdot}{\cdot}$ term is improved from $\ORS{n}{\frac{1}{15^k} \cdot \eps^{O(1/\eps)} \cdot n}$ to $\ORS{n}{\frac{1}{15^k} \cdot \poly(\eps) \cdot n}$.

% It is known that 

% \begin{align}
% 	n^{\Omega_{\delta}(1/\log\log{n})} \leq \ORS{n}{\delta n} \leq \frac{n}{\log^{(\text{poly}(1/\delta))}{(n)}}.
% \end{align}

% \noindent (See \cite{FischerLNRRS02,GoelKK12} and \cite{BG24} for the first and second inequalities, respectively.)

\subsection{Dynamic approximate matching via online matrix-vector multiplication}

\cite{Liu24} showed new algorithms and hardness results for dynamic $(1+\eps)$-matching when the input graph is bipartite.
The following shows that their results can be extended to general graphs using the new framework.
Liu's results are based on the connection between dynamic bipartite matching and the online matrix-vector problem ($\OMv$), which we present as follows.

\begin{definition}[$\OMv$ problem]
\label{def:omv}
In the $\OMv$ problem, an algorithm is given a Boolean matrix $M \in \{0,1\}^{n \times n}$. After preprocessing, the algorithm receives an online sequence of query vectors $v^{(1)}, \dots, v^{(n)} \in \{0,1\}^n$. After receiving $v^{(i)}$, the algorithm must respond the vector $Mv^{(i)}$.
\end{definition}

\begin{definition}[Dynamic approximate $\OMv$]
\label{def:dyn-approx-omv}
In the $(1-\lambda)$-approximate dynamic $\OMv$ problem, an algorithm is given a matrix $M \in \{0, 1\}^{n \times n}$, initially $0$. Then, it responds to the following:
\begin{itemize}
    \item $\textsc{Update}(i, j, b)$: set $M_{ij} = b$.
    \item $\textsc{Query}(v)$: output a vector $w \in \{0,1\}^n$ with $d(Mv, w) \le \lambda n$.
\end{itemize}
\end{definition}

Based on our framework (\cref{prob:framework,thm:framework}), to obtain algorithms for dynamic $(1+\eps)$-approximate matching, it suffices to describe how to handle the updates and queries in \cref{prob:framework}.
In our algorithms in this section, we handle the queries by implementing $\Aweak$ (\cref{def:Aweak}), in which the parameter $\lambda$ is fixed as a constant.

\subsubsection{Connection between dynamic approximate $\OMv$ and dynamic approximate matching} \label{sec:Liu-subsection}

This section aims to extend the following theorem to general graphs.

\begin{theorem}[{\cite[Theorem 2]{Liu24}}] \label{thm:equiv-bip}
There is an algorithm solving dynamic $(1-\lambda)$-approximate $\OMv$ with $\lambda = n^{-\sigma}$ with amortized $n^{1-\sigma}$ for \textsc{Update}, and $n^{2-\sigma}$ time for \textsc{Query}, for some $\sigma > 0$ against adaptive adversaries, if and only if there is a randomized algorithm that maintains a $(1-\eps)$-approximate dynamic matching with amortized time $n^{1-c} \peps$ in a bipartite graph, for some $c > 0$ against adaptive adversaries.
\end{theorem}

Since dynamic bipartite matching is a special case of dynamic matching, the if-direction of \cref{thm:equiv-bip} holds even on general graphs.
To generalize the only-if direction, we present an algorithm for dynamic $(1+\eps)$-approximate matching that assumes access to a dynamic $(1 - \lambda)$-approximate $\OMv$ algorithm.

\paragraph{Algorithm description.}
Assume that there is an algorithm $\AOMv$ for $(1 - \lambda)$-approximate dynamic matrix-vector for some $\lambda = n^{-\sigma}$ with update time $\cT = n^{1 - \sigma}$ and query time $n \cT = n^{2-\sigma}$.
Based on our framework (\cref{prob:framework} and \cref{thm:framework}), it suffices to describe how to handle the update chunks and queries.
In addition to $G$, we maintain its corresponding graph $B$ as defined in \cref{def:B}.
We remark that $B$ is used here for a different purpose, which is unrelated to our semi-streaming algorithm.
For a vertex subset $S \subseteq V$, let $S^+$ (resp. $S^-$) stands for $\{v^+ \mid v \in S\}$ (resp. $\{v^- \mid v \in S\}$).
The graph $B$ satisfies the following.

\begin{lemma} \label{lem:B}
For any vertex subset $S \subseteq V(G)$, we have $\mu(G[S]) \leq \mu(B[S^+ \cup S^-])$. In addition, any matching $M$ in $B[S^+ \cup S^-]$ can be transformed in $O(n)$ time into a matching in $G[S]$ of size at least $|M| / 6$.
\end{lemma}
\begin{proof}
For any matching $M$ in $G$, the edge set $\{(u^+, v^-) \mid (u, v) \in M\}$ is a matching in $B$.
Hence, $\mu(B) \geq \mu(G)$.

Let $M_B$ be a matching in $B$.
We can construct a subset of edges $X = \{(u,v) \mid (u^+, v^-) \in M_B\} \cup \{(u,v) \mid (u^-, v^+) \in M_B\}$ in $G$.
Note that every vertex has degree at most $2$ in $X$;
that is, each connected component is either a path or a cycle.
We can construct a matching $M_X$ by picking every other edge in the paths and cycles formed by $X$.
It is not hard to see that $|M_X| \geq |X| / 3$.
Since each edge in $X$ corresponds to at most 2 edges in $M_B$, we have $|X| \geq |M_B| / 2$.
Therefore, we can transform $M_B$ into the matching $M_X$, whose size is at least $|M_B| / 6$.
Since $M_B$ contains $O(n)$ edges, the transformation can be done in $O(n)$ time.
This completes the proof.
\end{proof}

We use $\AOMv$ to maintain the adjacency matrix $N_B$ of $B$.
That is, $N_B$ is initially empty, and whenever there is an edge update $(u, v)$ on $G$, we invoke the \textsc{update} operation to change entries for $(u^+, v^-)$ and $(v^+, u^-)$.

To handle queries, we implement $\Aweak$ using a lemma from \cite{Liu24}, which is for finding matchings in induced subgraphs of a bipartite graph.

\begin{lemma}[{\cite[Lemma 2.12]{Liu24}}] \label{lem:bip-Aweak}
Assume that we have access to an algorithm for dynamic $(1-\lambda)$-approximate $\OMv$ that maintains the adjacency matrix of a bipartite graph $B = (V_B = L \cup R, E_B)$ with query time $n\cT$. Then, there is a randomized algorithm that on vertex subsets $S_1 \subseteq  L, S_2 \subseteq R$, returns a $(2, \tO(D \lambda n))$-approximate matching of $B[S_1 \cup S_2]$ in $\tO(n^2/D + D n \cT)$ time for any parameter $D$, with high probability.
\end{lemma}

\begin{remark}
The algorithm in \cite[Lemma 2.12]{Liu24} actually returns a matching $M$ in $G'[A \cup B]$ such that it is almost maximal; i.e., at most $\tO(D \lambda n)$ edges in $G'$ are not adjacent to any edge in $M$.
Since a maximal matching is a 2-approximate matching, the returned matching is also a $(2, \tO(D \lambda n))$-approximation.
\end{remark}

Our implementation of $\Aweak$ is as follows.
Recall that each query $(S, \delta)$ requires finding a matching of size $\lambda \delta n$ in $G[S]$ for some constant $\lambda$, assuming that $\mu(G[S]) \geq \delta n$.
Upon receiving a query, we invoke the algorithm in \cref{lem:bip-Aweak} on $B$ with $S_1 = S^+$, $S_2 = S^-$, and $D = n^{\sigma/2}$ to obtain a matching $M_B$.
By \cref{lem:B}, if $\mu(G[S]) \geq \delta n$, then $\mu(B[S^+ \cup S^-]) \geq \delta n$.
Hence, the size of $M_B$ is at least
\[
\mu(B[S^+ \cup S^-]) / 2 - D \lambda n \geq \delta n / 2 - n^{1 - \sigma/2} \geq \delta n / 4,
\]
where the last inequality holds for a large enough $n$.
By \cref{lem:B}, $M_B$ can be transformed into a matching $M_G$ in $G[S]$ of size $\delta n / 24$.
Our implementation of $\Aweak$ returns the matching $M_G$.
Note that it satisfies \cref{def:Aweak} with $\lambda = 1/24$.

\paragraph{Time complexity.}
For each edge update (for \cref{prob:framework}), we spend $O(\log n)$ time to maintain $G$ and $B$;
also, we invoke \textsc{Update} of $\AOMv$, which takes $O(n^{1 - \sigma})$ amortized time.
Therefore, each edge update is handled in $O(n^{1 - \sigma} + \log n)$ amortized time.

For each query, computing $S^+$ and $S^-$ takes $O(n)$ time.
Then, we invoke \cref{lem:bip-Aweak} with $D = n^{\sigma/2}$, which takes $\tO(n^2 / D + D n \cT) = \tO(n^{2 - \sigma / 2})$ time.
By \cref{lem:B}, the final matching $M_G$ can be obtained in $O(n)$ time.
Therefore, each query is handled in $\tO(n^{2 - \sigma / 2})$ time.
Since there are $\peps$ queries every $\Theta(\eps^2 n)$ edge updates, the amortized time for the queries is $\peps \cdot \tO(n^{2 - \sigma / 2}) / (\eps^2 n) = \tO(\peps \cdot n^{1 - \sigma / 2})$.
Combining with the amortized time for handling updates (i.e. $O(n^{1 - \sigma} + \log n)$), our algorithm for \cref{prob:framework} has an amortized update time of $\tO(\peps n^{1 - \sigma/2})$.
By \cref{thm:framework}, this implies an algorithm for dynamic $(1+\eps)$-approximate matching with amortized update time $\tO(\peps n^{1 - \sigma/2})$.
Hence, we obtain the following.

\begin{theorem} \label{thm:equiv-general}
There is an algorithm solving dynamic $(1-\lambda)$-approximate $\OMv$ with $\lambda = n^{-\sigma}$ with amortized $n^{1-\sigma}$ for \textsc{Update}, and $n^{2-\sigma}$ time for \textsc{Query}, for some $\sigma > 0$ against adaptive adversaries, if and only if there is a randomized algorithm that maintains a $(1-\eps)$-approximate dynamic matching (in general graphs) with amortized time $\tO(\peps n^{1-\delta / 2})$ against adaptive adversaries.
\end{theorem}

\begin{remark}
As already mentioned in \cite{assadi2024improved}, the conditional lower bounds for dynamic approximate matching in \cref{thm:equiv-bip,thm:equiv-general} (which assume that dynamic $(1-\lambda)$-approximate $\OMv$ is hard) do not rule out an algorithm that only runs in $n^{o(1)}$-time when $\eps$ is a constant.
Therefore, it is possible that the algorithm in \cref{thm:AKS-dynamic} runs in $n^{o(1)}$ time while dynamic $\OMv$ is hard.
\end{remark}

\subsubsection{Faster algorithm for dynamic $(1+\eps)$-approximate matching in general graphs}
Using \cref{thm:equiv-general}, we can obtain an algorithm with $n / 2^{\Omega(\sqrt{\log n})}$ amortized update time for dynamic $(1+\eps)$-approximate matching in general graphs.
The algorithm is based on the following result.
Define the \emph{dynamic $\OMv$ problem} as the special case of dynamic $(1-\lambda)$-approximate $\OMv$ (\cref{def:dyn-approx-omv}) with $\lambda = 0$;
i.e. the output vector $Mv$ for each query must be computed without any approximation error.

\begin{lemma}[{\cite[Corollary 2.14]{Liu24}}]
\label{lem:dyn-omv}
There is a randomized algorithm for dynamic $\OMv$ against adaptive adversaries with amortized update time $n/2^{\Omega(\sqrt{\log n})}$ and query time $n^2/2^{\Omega(\sqrt{\log n})}$.
\end{lemma}

\noindent By letting $\sigma = \log_n(2^{\Omega{\sqrt{\log n}}})$, \cref{lem:dyn-omv} can be equivalently rephrased as a dynamic $\OMv$ algorithm with amortized update time $n^{1 - \sigma}$ and query time $n^{2 - \sigma}$.
Since the algorithm is exact, it is also a dynamic $(1 - \lambda)$-approximate $\OMv$ algorithm, where $\lambda = n^{-\sigma}$.
We obtain the following by combining this algorithm with \cref{thm:equiv-general}.

\begin{theorem} \label{thm:fast-dynamic}
There is a randomized algorithm that maintains a $(1+\eps)$-approximate maximum matching on a dynamic graph $G$ in amortized $\peps \cdot \frac{n}{2^{\Omega(\sqrt{\log n})}}$ update time against adaptive adversaries.
\end{theorem}

\subsubsection{Faster algorithm for offline dynamic $(1+\eps)$-approximate matching in general graph}
This section presents a faster algorithm for the \emph{offline dynamic $(1+\eps)$-approximate matching problem} in general graphs.
The algorithm adapts \cite{Liu24}'s algorithm for offline dynamic $(1+\eps)$-approximate matching on bipartite graphs. 
The problem is the same as dynamic $(1+\eps)$-approximate matching problem, except that the algorithm receives the whole sequence of edge updates as the input.
Therefore, our framework (\cref{prob:framework}) and \cref{thm:framework} are still applicable.
Furthermore, we can relax the problem as follows.
Let $G_i$ denote the graph $G$ after the first $i - 1$ chunks of edge updates.
Recall that in the reduction from dynamic matching to \cref{prob:framework} (see the proof of \cref{thm:framework}), each update chunk corresponds to a consecutive subsequence of edge updates, and each query corresponds to an invocation of $\Aweak$ for computing an approximate matching in $G_i$.
Since the edge updates are given offline, we can assume that all update chunks are given in the input;
in addition, the computation for all $G_i$ can be done simultaneously.
(However, the queries for a fixed $G_i$ are still given adaptively.)
We divide the computation for all $G_i$ into $\peps$ iterations, where in the $j$-th iteration, we handle the $j$-th query for graphs $G_1, G_2, \dots, G_t$ simultaneously; 
here, $t > 0$ is the number of chunks we handle simultaneously, which will be fixed later.

We need the following lemma to process the queries.

\begin{lemma}[{\cite[Lemma 2.11]{Liu24}}]
\label{lem:bip-parallel-Aweak}
Let $B_1, \dots, B_t$ be bipartite graphs on the same vertex set $V' = L \cup R$ such that $B_i$ and $B_1$ differ in at most $\Gamma$ edges.
Let $X_i \subseteq L, Y_i \subseteq R$ for $i \in [t]$.
There is a randomized algorithm that returns a maximal matching on each $B_i[X_i, Y_i]$ for $i \in [t]$ with high probability in total time
\[ \tO\left(t\Gamma + n^2t/D + D \cdot T(n, n/D, t) \right), \]
for any positive integer $D \leq n$.
\end{lemma}

As in \cref{sec:Liu-subsection}, we maintain the bipartite graph $B$ associated with $G$.
Let $B_i$ be the graph $B$ after the first $i - 1$ chunks of edge updates.
Note that all $B_i$-s are on the same vertex set $L \cup R$, where $L = V^+$ and $R = V^-$.
Consider queries $(S_1, \delta), (S_2, \delta), \dots, (S_t, \delta)$, where the $i$-th query is for the graph $G_i$.

\begin{lemma}
\label{lem:general-parallel-Aweak}
There is a randomized algorithm that returns a constant approximate matching on each $G_i[S_i]$ for $i \in [t]$ with high probability in total time
\[ \tO\left(t\Gamma + n^2t/D + D \cdot T(n, n/D, t) \right), \]
for any positive integer $D \leq n$.
\end{lemma}
\begin{proof}
We repeat the argument in the proof of \cref{thm:equiv-general} for this lemma.
That is, we first invoke \cref{lem:bip-parallel-Aweak}, with $X_i = S_i^+$ and $Y_i = S_i^-$ for all $i$, to find maximal matching for each $B_i[X_i \cup Y_i]$.
Let $M_i$ be the returned matching for $B_i$.
Since $M_i$ is maximal, it is a 2-approximate matching in $B_i[X_i \cup Y_i]$.
Using \cref{lem:B}, each $M_i$ is transformed into a matching $M'_i$ in $G_i[S_i]$.
Since $M_i$ is a 2-approximate matching, it follows from \cref{lem:B} that $M'_i$ is a 12-approximate matching in $G_i[S_i]$.

The time complexity is analyzed as follows.
The computation of $(X_1, Y_1), (X_2, Y_2), \dots, (X_t, Y_t)$ takes $O(tn)$ time.
Invoking \cref{lem:bip-parallel-Aweak} requires $\tO\left(t\Gamma + n^2t/D + D \cdot T(n, n/D, t) \right)$ time.
Transforming all $M_i$-s to $M'_i$-s takes $O(tn)$ time.
Since the parameter $D \leq n$, $tn = O(tn^2)$.
Hence, the overall time complexity is $\tO\left(t\Gamma + n^2t/D + D \cdot T(n, n/D, t) \right)$.
This completes the proof.
\end{proof}

Since there are $\peps$ queries for each $G_i$, by \cref{lem:general-parallel-Aweak}, we can handle all queries for $G_1, G_2, \dots, G_t$ in $\tO\left(t\Gamma + n^2t/D + D \cdot T(n, n/D, t) \right) \cdot \peps$ time.
Let $t = n^x$ and $D = n^y$ for some $x, y \in [0, 1]$ chosen later.
The amortized time complexity, over $\Gamma = t \cdot \Theta(\eps^2 n)$ edge updates, is
\begin{align*}
& \tO\left(\peps \cdot (t\Gamma + n^2t/D + D \cdot T(n, n/D, t)) / (tn) \right) \\
=~& \tO\left(\peps \cdot (t + n/D + D \cdot T(n, n/D, t) / (tn)) \right) \\
=~& \tO\left(\peps \cdot (n^x + n^{1 - y} + n^{y - 1 - x} \cdot T(n, n^{1-y}, n^x)) \right).
\end{align*}
For the choices $x = 0.579, y = 0.421$, the amortized runtime is $O(\peps \cdot n^{0.58})$.
(This choice of $(x, y)$ is found by \cite{Liu24}.)

\begin{theorem}
\label{thm:offline}
There is a randomized algorithm that given an offline sequence of edge insertions and deletions to an $n$-vertex graph (not necessarily bipartite), maintains a $(1+\eps)$-approximate matching in amortized $O(\peps \cdot n^{0.58})$ time with high probability.
\end{theorem}

% Consequently, solving offline dynamic $(1+\eps)$-approximate matching reduces to the following setup:
% We have graphs $G_1, \dots, G_t$, all with the same vertex set $V$ on $n$ vertices.
% Here, $t$ is the number of update chunks we handle simultaneously, which will be fixed later.
% Each $G_i$ differs from $G_1$ in at most $\Gamma = \Theta(t \cdot \eps^2 n)$ edges.
% Our goal is to implement $\Aweak$ as in \cref{def:Aweak} for all the graphs $G_1, G_2, \dots, G_t$ simultaneously.

\section*{Acknowledgments}
We are highly grateful to Sepehr Assadi for initiating this project. This work is supported by NSF Faculty Early Career Development Program No.~2340048. Part of this work was conducted while S.M.~was visiting the Simons Institute for the Theory of Computing.

In the original submission, our first result in \cref{table:times-dynamic} (\cref{thm:AKS-dynamic}) was incorrectly stated as $\poly(1/\eps) \cdot n^{o(1)} \cdot \ORS{n}{\Theta_\eps(n)}$. The correct result is $(1/\eps)^{O(1/\beta)} \cdot n^\beta \cdot \ORS{n}{\Theta_\eps(n)}$ for any given real number $\beta > 0$.
When $\beta$ is a constant, our result has a polynomial dependence on $1/\eps$.
However, the dependence becomes super-polynomial for $\beta = o(1)$.
We are grateful to Jiale Chen for pointing out this error.

\bibliographystyle{alpha}
\bibliography{ref.bib}

\newpage
\appendix

\section{Implementation in MPC and \CONGEST}
\label{sec:app-framework-distr}
Once equipped with \cref{thm:framework-1}, our MPC and \CONGEST results follow almost directly. 
Next, we provide a few details on implementing those in MPC and \CONGEST.

Our main analysis of the framework concerns the number of invocations to a $\Theta(1)$-approximate MCM oracle.
However, after the oracle returns a matching, a few updates and ``cleaning'' steps need to be performed. Those steps include extending alternating paths, contracting blossoms, removing specific vertices from the graph, and simultaneously propagating information throughout many disjoint components.
As long as each component has size that fits in the memory of a machine in MPC, those operations can be performed in $O(1)$ time. We refer a reader to \cite{andoni2018parallel} for details on implementing such a procedure. This yields the following corollary of \cref{thm:framework-1}.
\begin{corollary}
    Given a graph $G$ on $n$ vertices and $m$ edges, let $T(n, m)$ be the number of rounds needed to compute a $\Theta(1)$-approximate MCM in MPC.
    Then, there exists an algorithm that computes a $1+\eps$ approximation of MCM in $O(T(n, m) \cdot \log(1/\eps) / \eps^7)$ many rounds.
\end{corollary}

Implementing the aforementioned clean-up procedures is slightly more involved in \CONGEST.
Nevertheless, if a component has size $k$, necessary methods can be implemented in $O(k)$ \CONGEST rounds. 
To see why it is the case, observe first that all the vertices belonging to a structure of $\alpha$ can send their small messages to $\alpha$ in $O(k)$ rounds.
Then, after aggregating the received information, $\alpha$ can propagate necessary information to all the other vertices in its structure.
Given that the maximum component size our algorithm ensures is $1/\eps^3$, we obtain the following corollary.
\begin{corollary}
    Given a graph $G$ on $n$ vertices and $m$ edges, let $T(n, m)$ be the number of rounds needed to compute a $\Theta(1)$-approximate MCM in \CONGEST.
    Then, there exists an algorithm that computes a $1+\eps$ approximation of \CONGEST in $O(T(n, m) \cdot \log(1/\eps) / \eps^{10})$ many rounds.
\end{corollary}
\section{Correctness of the simulation} \label{sec:app-correctness}

In this section, we prove the correctness of our simulation.
More precisely, we show that the procedures \algCheck and \algExtend in \cite{MMSS25}'s algorithm (see \cref{alg:outline} and \cref{alg:phase}) can be replaced by the corresponding simulated versions in \cref{sec:framework-distr}, and the resulting algorithm still outputs a $(1+\eps)$-approximate maximum matching.
We present the missing proofs in \cref{sec:app-missing-proof}.
An overview of the correctness proof is given in \cref{sec:app-actual-correctness}.
Our proofs in this section very closely follow those in \cite{MMSS25}.

\subsection{Missing proofs in \cref{sec:sim-extend-1}} \label{sec:app-missing-proof}

Recall that the analysis for \cref{alg:sim-extend-1} has three parts.
In the following, we present the proof of the first two parts.

\paragraph{Part 1.} 
Recall that Part 1 is to show that every \algOvertake operation in \cref{line:sim-ext-overtake} of \cref{alg:sim-extend-1} is performed on an $s$-feasible (and therefore type 3) arc. 
We first establish a property of $s$-feasible arcs.

\begin{lemma} \label{lem:feasible}
Let $(a', b'), (c', d') \in G'$ be two $s$-feasible arcs that do not share endpoints.
Suppose that we modify $G'$ by performing $\algOvertake$ on $(a', b')$.
Then, $(c', d')$ remains $s$-feasible after the modification.
\end{lemma}
\begin{proof}
Let $\cS_1, \cS_2, \cS_3, \cS_4$ be, respectively, the structures containing $a', b', c', d'$ before the $\algOvertake$ operation;
if $b'$ (resp. $d'$) is an unvisited vertex, we define $\cS_2$ (resp. $\cS_4$) as the single vertex $b'$ (resp. $d'$).
We first prove the following claim.

\begin{claim}\label{clm:small-claim} 
Before the overtaking operation, $\cS_1 \neq \cS_3$, and $b'$ were not on the active path of $\cS_3$.     
\end{claim}
\begin{proof}
Before the overtaking operation, $(a', b')$ and $(c', d')$ were of type 3.
Hence, $a'$ and $c'$ must be the working vertex of $\cS_1$ and $\cS_3$, respectively.
Since $(a', b'), (c', d')$ do not share endpoints, $a'\neq c'$ and thus $\cS_1 \neq \cS_3$.
In addition, since $\ell(a') = \ell(c') = s$, $b'$ cannot be on the active path of $\cS_3$.
This completes the proof.
We remark that it is possible for $\cS_1 = \cS_2$ (i.e., $a'$ is overtaking a vertex in its own structure) $\cS_2 = \cS_3$ (i.e., $a'$ is overtaking the structure of $c'$) or $\cS_2 = \cS_4$ (i.e., $a'$ and $c'$ are overtaking from the same structure).    
\end{proof}

\noindent We now show that $(c', d')$ remains $s$-feasible after the operation.
(See \cref{def:s-feasible} and \cref{def:arc-type} for the conditions.)
The $\algOvertake$ operation on $(a', b')$ consists of the following steps (see \cref{sec:overtake} for details):
\begin{itemize}
    \item Modify the alternating trees of $\cS_1$ and $\cS_2$ by re-assigning the parent of $b'$ as $a'$.
    \item Update the label of the matched edge incident to $b'$.
    \item Change the working vertex of $\cS_1$ and $\cS_2$.
    \item Mark $\cS_1$ as extended; mark $\cS_2$ as overtaken.
\end{itemize}

\noindent Since the overtaken node $b'$ was not on the active path of $\cS_3$ (by \cref{clm:small-claim}), all nodes on the active path of $\cS_3$ is unchanged.
Since $b' \neq d'$, $\ell(d')$ is also unchanged.
Hence, $c'$ is still the working vertex of $\cS_3$, and it still holds that $\ell(c') = s$ and $\ell(d') > s$.
Since $\cS_1 \neq \cS_3$ (by \cref{clm:small-claim}), $\cS_3$ is not marked as extended.
Clearly, $d'$ is still an inner or unvisited vertex and $\cS_3$ is not on-hold.
This shows that $(c', d')$ is still $s$-feasible.
\end{proof}

Let $M' = \{e_1, e_2, \dots, e_k\}$ be the matching computed in \cref{line:sim-ext-Amat}.
Since $M'$ is a matching, $e_1, e_2, \dots, e_k$ do not share endpoints.
By \cref{lem:feasible}, after we perform \algOvertake on $e_1$, all other arcs in $M'$ remain $s$-feasible.
By induction, where we repeat the above argument in the inductive step, we show that $e_i, e_2, \dots, e_k$ are $s$-feasible after we perform \algOvertake on $e_1, e_2, \dots, e_{i-1}$, for all $i \geq 1$.
Therefore, every operation in \cref{line:sim-ext-overtake} of \cref{alg:sim-extend-1} is performed on an $s$-feasible arc.

% An implication of \cref{lem:feasible} is as follows.
% Let $S$ be a set of $s$-feasible arcs that form a matching in $G'$.
% If we perform an overtaking operation on any arc in $S$, the rest of the arcs are still $s$-feasible (and therefore of type 3).
% Hence, we may perform $\algOvertake$ on all arcs in $S$ to simulate the behavior of the semi-streaming algorithm when reading the arcs in $S$ one by one.

\paragraph{Part 2.}
Recall that the second part is to prove the following lemma.

\extcont*
\begin{proof}
Recall that we invoke the simulation of $\algCheck$ after completing the last stage.
Hence, every $G$-arc of type 1 or 2 is marked as contaminated.
Suppose that there is a non-contaminated arc $(u, v) \in G$ that is of type 3.
Let $u' = \Omg(u)$, $v' = \Omg(v)$, and $s = \ell(u')$.
In the following, we show that $(u, v)$ is already a type 3 arc at the end of stage $s$, which contradicts the assumption because $(u, v)$ should have been marked as contaminated at the end of statge $s$.

Since $(u, v)$ is of type 3, $\cS_{u}$ is not marked as extended or on-hold.
Hence, during the simulation, $u$ is not overtaken by any structure.
(That is, it stays in the structure of the same free vertex. However, it is possible that this structure is overtaken but the operation did not overtake $u$).
In addition, $\cS_u$ did not overtake, and we did not apply \algContract or \algAugment on $\cS_u$.
Therefore, \textbf{at the end of stage $s$, $u'$ is already the working vertex of $\cS_u$, and $\ell(u') = s$.}

Since $v'$ is not an outer vertex at the end of the simulation, it is not an outer vertex at the end of $s$.
Since the label of $v'$ can only decrease during the simulation, $\ell(v') > s + 1$ at the end of $s$.
By \cref{def:arc-type}, $(u, v)$ was a type 3 arc at the end of stage $s$.
This completes the proof.
\end{proof}

\subsection{Key ingredience of the correctness proof} \label{sec:app-actual-correctness}

The correctness proof of \cite{MMSS25}'s algorithm consists of two key ingredients: 

\begin{enumerate}
    \item[(I1)] Their algorithm does not miss any short augmentation.
    That is, if one phase of the algorithm is left to run indefinitely and no structure is on hold, then at some point, the remaining graph will have no short augmentation left.
    \item[(I2)] If there are at least $4h\lmax|M|$ vertex-disjoint augmenting paths in $G$, then the size of $M$ is increased by a factor of $1 + \frac{h \lmax}{\Delta_h}$ in this phase.
\end{enumerate}

(I2) shows that the algorithm outputs a $(1+\eps)$-approximate matching after running certain numbers of scales and phases.
(I1) highlights the intuition behind the algorithm and is used to prove (I2).

To prove the correctness, we show that (I1) and (I2) also hold for our simulation, except for a small difference caused by contaminated arcs.

Unless otherwise stated, all lemmas, corollaries, and theorems in this section refer to our simulation, in which \algCheck and \algExtend are replaced with \cref{alg:sim-check-1} and \cref{alg:sim-extend-1}.

\subsection{Proof of the first ingredient}
The modified version of (I1) is as follows.

% \begin{definition}[\algSim] \label{def:simulation}
% For ease of notation, we use \algSim to denote the procedure obtained by replacing \algCheck and \algExtend in \cite{MMSS25}'s algorithm (\cref{alg:outline}) with \cref{alg:sim-check-1} and \cref{alg:sim-extend-1}, respectively.
% \end{definition}

\begin{definition}[Critical arc and vertex]
\label{def:critical}
Recall that the active path of a structure is a path in $G'$.
We say a non-blossom arc $(u, v) \in G$ is \emph{critical} if the arc $(\Omg(u), \Omg(v)) \in G'$ is active.
In particular, all blossom arcs in a structure are not critical, even if they are in an active blossom.
We say a free vertex $\alp \in G$ is \emph{critical} if $S_\alp$ is active.
\end{definition}

\begin{theorem}[No short augmenting paths is missed] \label{lem:active}
At the beginning of each $\bundle$, the following holds.
Let $P = (\alp, a_1, a_2, \dots, a_k, \beta)$ be an augmenting path in $G$ such that no vertex in $P$ is removed in this phase and $k \leq \lmax$.
At least one of the following holds: $\alp$ is critical, $P$ contains a critical arc, or $P$ contains a contaminated arc. 
\end{theorem}

\paragraph{Properties of a phase.}
%\subsection{Properties of a phase}

%Consider a fixed phase with respect to a fixed scale $h$.
Our analysis of \cref{lem:active} relies on the following three simple properties.

\begin{corollary} \label{cor:contaminated}
After the execution of \algCheck in each \bundle, there are no arcs of type 1, 2, or 3 in $G$ except the contaminated ones.    
\end{corollary}

The following lemma is proven by \cite{MMSS25}.
It is not hard to check that their argument is not affected by contaminated arcs.

\begin{lemma}[Outer vertex has been a working one, \cite{MMSS25}]
\label{lem:working} 
Consider a $\bundle$ $\tau$.
Suppose that $G'$ contains an outer vertex $v'$ at the beginning of $\tau$.
Then, there exists a \bundle $\tau' \leq \tau$ such that $v'$ is the working vertex at the beginning of $\tau'$.
\end{lemma}

\begin{restatable}{invariant}{invouter}
\label{inv:outer-independence} At the beginning of each $\bundle$, no arc in $G$ connects two outer vertices, unless the arc is contaminated.
\end{restatable}

\begin{lemma}
\label{lem:outer-independence} 
\cref{inv:outer-independence} holds.
\end{lemma}

\noindent \cref{cor:contaminated} is a direct consequence of \cref{lem:extend-contaminated}.
Proofs of \cref{lem:outer-independence} are deferred to \cref{sec:lem:outer-independence}.

We remark that \cref{inv:outer-independence} only holds at the beginning of each $\bundle$.
During the execution of $\algExtend$, some structures may include new unvisited vertices or contract a blossom in the structure. These operations create new outer vertices that may be adjacent to existing ones.

The contaminated arcs can cause some augmenting paths to be missed.
We ensure that (I2) holds by running more $\bundle$s

\subsubsection{No short augmentation is missed (Proof of \cref{lem:active})}

For a $\bundle\ \tau$, we use $\Omg^\tau$ and $\ell^\tau$ to denote, respectively, the set of blossoms and labels at the beginning of $\tau$.

%\begin{proof} 
Consider a fixed phase and a $\bundle$ $\tau$ in the phase.
Suppose, toward a contradiction, that at the beginning of $\tau$, there exists an augmenting path $P = (\alp, a_1, a_2, \dots, a_k, \beta)$ in $G$, where $k \leq \lmax$, such that:
\begin{enumerate}[(i)]
    \item none of the vertices in $P$ is removed in the phase,
    \item $\alp$ and all arcs in $P$ are not critical, and
    \item all arcs in $P$ are not contaminated.
\end{enumerate}
Recall that for an arc to be critical, by \cref{def:critical}, it has to be non-blossom. Hence, some blossom arcs of $P$ may be inside of an active blossom at this moment.
For $i = 1, 2, \dots, k$, let $u_i$ and $v_i$ be the tail and head of $a_i$, respectively; that is, $a_i = (u_i, v_i)$.
Let $v_0 = \alp$.
Two cases are considered:

%\begin{itemize}
%    \item[] \textbf{Case 1:} 
    
\paragraph{Case 1: There exists an index $q$ such that $\ell^{\tau}(a_q) > q$.}   

    Let $q$ be the smallest index such that $\ell^{\tau}(a_q) > q$.
    Since $\ell^{\tau}(a_q) > q > 0$, $a_q$ must be a non-blossom arc.
    Let $p$ be the smallest index such that $p < q$ and $a_{p+1}, \dots, a_{q-1}$ are blossom arcs.
    
    We first show that $\Omg(v_p)$ is an outer vertex, and all vertices in the path $(v_p, a_{p+1}, a_{p+2}, \dots, a_{q-1})$ are in the same inactive blossom.
    If $p = 0$, then $v_p = \alp$ and thus $\Omg(v_p)$ is an outer vertex.
    Otherwise, since $\ell^{\tau}(a_p) \leq p \leq \lmax$, $a_p$ is visited.
    Hence, $a_p$ is a non-blossom arc contained in a structure, which also implies that $\Omg(v_p)$ is an outer vertex.
    For $p < i < q$, since $a_i$ is a blossom arc, $\Omg(u_i)$ and $\Omg(v_i)$ are both outer vertices.
    Hence, by \cref{inv:outer-independence}, all vertices in the path $(v_p, a_{p+1}, a_{p+2}, \dots, a_{q-1})$ must be in the same blossom at the beginning of $\tau$.
    (Here, we can apply \cref{inv:outer-independence} because all arcs in $P$ are non-contaminated.)
    Denote this blossom by $B$.
    For $p > 0$, since $a_p$ is non-critical at the beginning of $\tau$, we have that \textbf{$B$ is inactive at the beginning of $\tau$}.
    For $p = 0$, since $v_p = \alp$ is non-critical, we also have $B$ is inactive at the beginning of $\tau$.

    Let $\tau' \leq \tau$ be the last $\bundle$ such that $B$ is the working vertex of some structure $S_\gamma$ at the beginning of $\tau'$.
    By \cref{lem:working}, $\tau'$ exists, and since $B$ is inactive at the beginning of $\tau$, we further know that $\tau'$ < $\tau$.
    In $\tau'$, $S_\gamma$ backtracks from $B$, and $B$ remains inactive until at least the beginning of $\tau$.
    Hence, if $p > 0$, $\ell(a_p)$ is not updated between the end of $\tau'$ and the beginning of $\tau$.
    Therefore, \textbf{$\ell^{\tau'}(a_p) = \ell^{\tau}(a_p) \leq p$}.

    By definition of $\algBacktrack$ (\cref{sec:backtrack}), $S_\gamma$ is not marked as on hold or modified in $\tau'$.
    Hence, $\Omg(v_{q-1}) = B$ is the working vertex of $S_\gamma$ during the whole \bundle $\tau'$.
    Consider the moment when we finish the execution of \algCheck in $\tau'$.
    By \cref{cor:contaminated}, the path $P$ cannot contain any arc of type 1, 2, or 3 at this moment.
    We now show that this leads to a contradiction.
    First, we claim that $\cev{a_q}$ cannot be in any structure.
    (Recall that $\cev{a_q}$ is the reverse direction of $a_q$.)
    If $\cev{a_q}$ is in any structure, then $\Omg(u_q)$ is an outer vertex.
    This implies that $g$ is an arc of either type 1 or 2, a contradiction.
    Hence, $\cev{a_q}$ is not in any structure, and $\Omg(u_q)$ is either unvisited or an inner vertex.
    In addition, $\ell^{\tau'}(a_q) > q \geq \ell^{\tau'}(a_p) + 1$.
    (Here, for ease of notation, we define $\ell^{\tau'}(a_p) = 0$ if $p = 0$.)
    This again leads to a contradiction, because $g$ is a type 3 arc (i.e. \algOvertake should have been performed on $g$).

    % 10/14 modified: this part is removed due to the simplification
    % {\color{red}Recall that $v_p$ is the base of $B$.
    % Thus, $\lab(u)$ is computed by finding the shortest even-length alternating path $Q^*$ in $E(B) \cap E(G_\gamma)$ from $v_p$ to $v_{q-1}$.
    % By \cref{inv:complete}, $G_\gamma$ contains all blossom arcs in $E(B)$ at the beginning of $\tau'$.
    % Since $S_\gamma$ is not marked as modified in the $\bundle$, $G_\gamma$ still contains these arcs when $Q^*$ is computed.
    % Hence, all arcs in the path $(v_p, a_{p+1}, \dots, a_{q-1})$ are stored in $G_\gamma$ at this moment.
    % As a result, the length of $Q^*$ is at most the length of $(v_p, a_{p+1}, \dots, a_{q-1})$.
    % If $p = 0$, $\lab(u)$ is computed as the number of matched arcs in $Q^*$, which is at most $q - p - 1 = q - 1$; otherwise, $\lab(u)$ is computed as $\ell^{\tau'}(a_p) + (q - p - 1) \leq q - 1$, where we use our conclusion from above that $\ell^{\tau'}(a_p) = \ell^{\tau}(a_p) \leq p$.
    % This implies that $\algExtend$ should have updated the label of $a_q$ to $q$ in $\tau'$, which contradicts the definition of $q$.}
        
    %\item[] \textbf{Case 2:} 
    
    \paragraph{Case 2: For each $i = 1, 2, \ldots, k$ it holds $\ell^{\tau}(a_i) \leq i$.}
    Let $p \leq k$ be the smallest index such that all vertices in $a_{p+1}, a_{p+2}, \dots, a_k$ are blossom arcs at the beginning of $\tau$. 
    Similar to Case 1, if $p = 0$, then $v_p = \alp$; otherwise, $a_p$ is a non-blossom arc with $\ell^{\tau}(a_p) \leq p \leq \lmax$. 
    Hence, $\Omg(v_p)$ is an outer vertex.

    For each $i > p$, since $a_i$ is a blossom arc, $\Omg(u_i)$ and $\Omg(v_i)$ are both outer vertices.
    Hence, by \cref{inv:outer-independence}, all vertices in the path $(v_p, a_{p+1}, a_{p+2}, \dots, a_{k})$ must be in the same blossom at the beginning of $\tau$.
    Denote this blossom by $B$.
    Since $\Omg^\tau(\beta)$ is the root of $T'_\beta$, it is also an outer vertex.
    That is, $\Omg^\tau(v_k)$ and $\Omg^\tau(\beta)$ are both outer vertices.
    
    Since $P$ contains an unmatched arc $(v_k, \beta)$, by \cref{inv:outer-independence}, $\Omg^\tau(\beta) = \Omg^\tau(v_k) = B$.
    If $p = 0$, this leads to a contradiction because $B$ contains two free vertices $v_p = \alp$ and $\beta$. If $p > 0$, then $B = \Omg(\beta)$ is not a free vertex because it is adjacent to a non-blossom matched arc $a_p$, which is also a contradiction. 
%\qedhere
%\end{itemize} 
%\end{proof}

\subsubsection{No arc between outer vertices (Proof of \cref{lem:outer-independence})}
\label{sec:lem:outer-independence}
%\begin{proof}
Suppose, by contradiction, that at the beginning of some $\bundle$ $\tau$, there is a non-contaminated arc $g' \in E(G')$ connecting two outer vertices $u'$ and $v'$.
Let $\tau_{u'} < \tau$ (resp. $\tau_{v'} < \tau$) be the first $\bundle$ in which $u'$ (resp. $v'$) is added to a structure by either $\algOvertake$ or $\algContract$. Without loss of generality, assume that $\tau_{u'} \geq \tau_{v'}$. There are two cases:
\begin{enumerate}
    \item $u'$ is added to a structure during \algExtend (by either \algOvertake or \algContract).
    \item $u'$ is added to a structure during \algCheck.
\end{enumerate}

In the first case, by the proof of \cref{lem:working}, we know that $u'$ is a working vertex after the execution of $\algExtend$.
That is, $g'$ is of type 1 or 2.
This contradicts \cref{lem:extend-contaminated}.

In the second case, $g'$ is of type 1 or 2 at the moment it is added to a structure.
If it is not contracted or removed, it is still of type 1 or 2 after \algCheck.
This contradicts \cref{cor:contaminated}.

\subsection{Proof of the second ingredient}

\cite{MMSS25} showed the following; their argument is correct even with contaminated arcs.

\begin{lemma}[Upper bound on the number of active structures, \cite{MMSS25}] \label{lem:active-bound} 
Consider a fixed phase of a scale $h$. Let $M$ be the matching at the beginning of the phase. Then, at the end of that phase, there are at most $h|M|$ active structures.
\end{lemma}

The second ingredient is proven as follows.

\begin{lemma} \label{lem:path-bound} Consider a fixed phase in a scale $h$. Let $M$ be the matching at the beginning of the phase. Let $\calP^*$ be the maximum set of disjoint $M$-augmenting $\lmax$-paths in $G$. If $|\calP^*| \geq 4h\lmax |M|$, then at the end of the phase the size of $M$ is increased by a factor of at least $1 + \frac{h\lmax}{\Delta_h}$.
\end{lemma}
\begin{proof}
Consider an augmenting path $P$ in $\calP^*$.
By \cref{lem:active}, at the end of the phase, one of the following must hold:
\begin{enumerate}
    \item $P$ contains a critical arc or critical vertex,
    \item Some vertices in $P$ are removed, or
    \item $P$ contains a contaminated arc.
\end{enumerate}

Let $\calP$ be the set of disjoint $M$-augmenting paths found in this phase.
In \cite{MMSS25}, it has been shown that there are at most $h|M| \cdot (\lmax + 1) + 2\Delta_h \cdot |\calP|$ paths in $\calP^*$ containing a removed vertex, critical vertex, or critical arc.

We now bound the number of paths in $\calP^*$ containing a contaminated arc.
By \cref{lem:vc,lem:vc2}, the number of contaminated arcs marked in a \bundle is at most $4\eps^{17} |M|$.
Since each phase has $\frac{72}{h\eps}$ $\bundle$s, the total number of contaminated arcs in a phase is $288\eps^{16} |M| / h \leq 18432 \eps^{14} |M| = \Theta(\eps^{14} |M|)$, where the inequality comes from $h \geq \frac{\eps^2}{64}$.
For small enough $\eps$, this number is at most $h|M|$ (recall that $h|M| = \Omg(\eps^2|M|)$).

Combine all above, we obtain $|\calP^*| \leq h|M|(\lmax + 1) + h|M| + |\calP| \cdot 2\Delta_h$.
Assume that $|\calP^*|$ contains more than $4h|M|\lmax$ paths, then we have
\[
    |\calP|
    \geq \frac{|\calP^*| - h|M|(\lmax + 2)}{2\Delta_h}
    \geq \frac{2h\lmax}{2\Delta_h}|M|
    = \frac{h\lmax}{\Delta_h}|M|.
\]
The algorithm augments $M$ by using the augmenting paths in $\calP$ at the end of the phase.
Hence, the size of $M$ is increased by a factor of $(1 + \frac{h\lmax}{\Delta_h})$ at the end of this phase.
This completes the proof.
\end{proof}

The above lemma shows that each phase increases the size of $M$ by a factor, and the factor is exactly the same as the one proven in \cite{MMSS25}'s original paper despite the presence of contaminated arcs.
(Intuitively, this is because the analysis in \cite{MMSS25} is not tight, and the number of contaminated arcs is so small that its effect fits into the slack of the analysis.)
By the analysis in \cite{MMSS25}, the algorithm still outputs a $(1+\eps)$-approximation after the chosen number of scales and phases.

\end{document}